\documentclass[%
 reprint,
superscriptaddress,
 amsmath,amssymb,
 aps,
prb,
floatfix,
]{revtex4-2}

\newcommand{\R}{\mathbb{R}}

\DeclareMathOperator{\rank}{rank}

\newcommand{\Xx}{\vec{X}}
\newcommand{\Yy}{\vec{Y}}

\renewcommand{\to}{\longrightarrow}
\renewcommand{\bar}{\overline}
\renewcommand{\hat}{\widehat}

\newenvironment{illustration}{\par\noindent\textit{Illustration.}\ }{\hfill$\square$\par}

\usepackage{hyperref}

\usepackage[T1]{fontenc}
\usepackage[utf8]{inputenc}
\usepackage{microtype}
\usepackage{lipsum}
\usepackage{soul}
\usepackage{amsmath}
\usepackage{amssymb}
\usepackage{amsfonts}
\usepackage{amsthm}
\usepackage{mathtools}
\usepackage{latexsym}
\usepackage{bbm}
\usepackage{braket}

\usepackage{algorithm}
\usepackage{algpseudocode}

\algblock{ParFor}{EndParFor}
\algnewcommand\algorithmicparfor{\textbf{parallel for}}
\algnewcommand\algorithmicpardo{\textbf{do}}
\algnewcommand\algorithmicendparfor{\textbf{end parallel for}}
\algrenewtext{ParFor}[1]{\algorithmicparfor\ #1\ \algorithmicpardo}
\algrenewtext{EndParFor}{\algorithmicendparfor}

\usepackage{graphicx}
\usepackage{tikz}
\usetikzlibrary{arrows.meta,calc,positioning,decorations.pathreplacing}

\usepackage{tikz}
\usepackage{pgfplots}
\usepackage{pgfplotstable}
\usepgfplotslibrary{groupplots}
\pgfplotsset{compat=1.18}

\definecolor{oiBlack}      {RGB}{  0,  0,  0}
\definecolor{oiOrange}     {RGB}{230,159,  0}
\definecolor{oiSkyBlue}    {RGB}{ 86,180,233}
\definecolor{oiBluishGreen}{RGB}{  0,158,115}
\definecolor{oiYellow}     {RGB}{240,228, 66}
\definecolor{oiBlue}       {RGB}{  0,114,178}
\definecolor{oiVermillion} {RGB}{213, 94,  0}
\definecolor{oiPurple}     {RGB}{204,121,167}

\makeatletter
\pgfplotsset{
  discard if not/.style 2 args={
    x filter/.append code={
      \edef\pgfplots@tempa{\thisrow{#1}}%
      \edef\pgfplots@tempb{#2}%
      \ifx\pgfplots@tempa\pgfplots@tempb\else
      \fi
    }
  }
}
\makeatother

\pgfplotsset{
  nqsaxis/.style={
    width               = 0.415\textwidth,
    height              = 0.315\textwidth,
    scale only axis,
    axis line style     = {black, line width=0.4pt},
    tick style          = {black, line width=0.4pt},
    tick align          = outside,
    tick pos            = left,
    major tick length   = 2.5pt,
    minor tick length   = 1.3pt,
    label style         = {font=\sffamily\small},
    tick label style    = {font=\sffamily\footnotesize},
    grid                = none,
    clip                = true,
    filter discard warning = false,
    unbounded coords    = discard,
  },
  nqspoint/.style={only marks, mark size=1.85pt, line width=0.45pt},
  nqspointopen/.style={only marks, mark size=2.20pt, line width=0.55pt},
  nqsline/.style={mark size=1.55pt, line width=0.55pt},
}

\usepackage[dvipsnames]{xcolor}
\usepackage{tcolorbox}

\usepackage[
]{hyperref}
\hypersetup{pdfpagemode=UseNone}

\newtheorem{theorem}{Theorem}
\newtheorem{lemma}{Lemma}

\theoremstyle{definition}
\newtheorem{definition}{Definition}

\theoremstyle{remark}

\begin{document}

\title{Symmetry Constraints Regularize Neural Quantum State Learning}

\author{Turbasu Chatterjee}
\thanks{These authors contributed equally to this work.}
\affiliation{Department of Electrical and Computer Engineering, North Carolina State University, Raleigh, North Carolina, USA}

\author{Manas Sajjan}
\thanks{These authors contributed equally to this work.}
\affiliation{National Center for Computational Sciences, Oak Ridge National Laboratory, Oak Ridge, Tennessee 37830, United States}

\author{Songbo Xie}
\affiliation{Department of Electrical and Computer Engineering, North Carolina State University, Raleigh, North Carolina, USA}

\author{Elliott Love}
\affiliation{Department of Mathematics, North Carolina State University, Raleigh, North Carolina, USA}

\author{Vinit Singh}
\affiliation{Department of Electrical and Computer Engineering, North Carolina State University, Raleigh, North Carolina, USA}

\author{Bojko N. Bakalov}
\email{bnbakalo@ncsu.edu}
\affiliation{Department of Mathematics, North Carolina State University, Raleigh, North Carolina, USA}
\author{Sabre Kais}
\email{skais@ncsu.edu}
\affiliation{Department of Electrical and Computer Engineering, North Carolina State University, Raleigh, North Carolina, USA}

\begin{abstract}
Neural quantum states (NQS) offer highly expressive variational wavefunctions, but their optimization is frequently bottlenecked by redundant parameters and poorly conditioned landscapes. We demonstrate that embedding Hamiltonian symmetries directly into the variational parameterization geometrically regularizes this learning problem. For Boltzmann-family NQS, we enforce symmetries by tying local Pauli-$Z$ generators along physical geometric orbits, analytically collapsing the trainable coefficient space prior to optimization. To quantify the resulting optimization geometry, we introduce a geometric metric built on the Jacobian and Hessian of the optimization landscape. This framework evaluates the fraction of the physically accessible state space that corresponds to high-quality, low-energy solutions. Evaluating our approach on transverse-field Ising (TFIM) and XXZ spin chains shows that symmetry compilation excises the vast majority of parameters while maintaining ground-state accuracy within the resolution of the reported benchmarks. In large TFIM systems, strong spatial constraints compress thousands of parameters down to tens, delivering substantial runtime accelerations. Our geometric diagnostics indicate that symmetry produces a more favorable target-aware geometry by concentrating the reachable state space around low-energy solutions while retaining broad target basins. Together, our results indicate that symmetry compilation concentrates the expressive power of NQS on states relevant to the target problem, thereby reducing model size and training cost without sacrificing accuracy.

\end{abstract}

\maketitle
\section{Introduction}
\label{sec:introduction}

Neural quantum states (NQS) provide a variational language for many-body wavefunctions~\cite{carleo2017solving,PhysRevX.7.021021,lange2024architecturesapplicationsreviewneural,vivas2022neuralnetworkquantumstatessystematic,medvidovic2024neural}, capable of capturing correlations that elude conventional tensor networks~\cite{sharir2022neural,choo2020fermionic,hibat2020recurrent,barrett2022autoregressive}. However, this expressivity introduces severe optimization challenges. Large, highly expressible NQS architectures are often plagued by problems closely related to barren plateaus studied in parameterized quantum circuits~\cite{mcclean2018barren,cerezo2021cost,holmes2022connecting,larocca2022diagnosing,larocca2024review,ragone2024lie}. These arise due to redundant gradient directions arising from weakly identifiable parameters~\cite{RBM_Geom_learning,PhysRevB.100.195125,mcclean2018barren,cerezo2021cost}, i.e., the same parameters that allow a model to be expressible, arrest its ability to converge quickly and minimize its local cost function.

For homogeneous models, symmetries, such as translations, point-group operations, or global spin inversions, offer a natural mechanism to prune this excess capacity~\cite{gard2020efficient,bravyi2017tapering,PhysRevResearch.3.013039,sauvage2024building,wiersema2024geometricquantummachinelearning,ragone2024lie,wiersema2024classification}. Numerical many-body methods enforce symmetries through Hilbert-space projection, symmetry-sector restriction, symmetry-preserving ansatze, or penalty constructions~\cite{bravyi2017tapering,gard2020efficient,PhysRevResearch.3.013039,choo2018,sauvage2024building}. These approaches can reduce the effective variational search space and improve optimization, motivating the following question: \textit{Beyond simply decreasing the number of parameters, in what specific ways does embedding symmetry into NQS models alter the structure of the optimization landscape and how can we quantify its impact using a metric that is both invariant under coordinate transformations and sensitive to the target state?}

In this paper, we address this by structurally embedding symmetries into Boltzmann-family NQS~\cite{salakhutdinov2009deep,salakhutdinov2010efficient,melko2019restricted,fischer2014training,amin2018quantum,PhysRevB.96.205152,RBM_anyons_symm,Topo_state_RBM,RBM_pruning,demidik2025expressiveequivalenceclassicalquantum} and studying its geometric loss landscape. By expanding the real amplitude and phase generators in local Pauli-$Z$ strings, the model is parameterized by coefficient vectors over a $k$-local support basis. Rather than penalizing asymmetry during optimization, we enforce invariance directly on these coefficient functions. By tying coefficients along physical geometric orbits, such as translations or full space-group representations, we collapse the exponentially large configuration space into symmetry-compatible eigenvalue classes prior to training, removing redundant parameter directions and leaving an optimization problem that is better conditioned from the start. This structural projection analytically dictates the parameter compression, linking dimensionality reduction directly to the intrinsic degeneracies of the target state.

However, naive parameter compression is an incomplete proxy for trainability~\cite{holmes2022connecting,nakaji2021expressibility,larocca2024review,jacot2018neural,lee2019wide,novak2019neural}. To quantify trainability and diagnose ill-conditioning, we compare the local volume of low-energy states with the total physical volume reachable by the ansatz. The Fubini--Study metric provides a coordinate-invariant measure of state-space volume~\cite{koczor2019quantum,haug2021capacity,beckey2022variational,meyer2021fisher,abbas2021effective,abbas2021power}, while the curvature of the energy landscape determines the leading geometry of the target basin. By evaluating the fraction of the physically accessible manifold occupied by these high-quality configurations, we show that we are able to isolate the genuine geometric regularization  under the hood that arises from coordinate shrinkage.

We test our hypotheses on the transverse-field Ising (TFIM) and XXZ spin chains~\cite{koffel2012entanglement,vitagliano2010volume,scheie2021detection}. Across both models, spatial and bitflip symmetry compilation structurally reduces the vast majority of parameters, delivering orders-of-magnitude wall-time speedups without sacrificing ground-state accuracy. The geometric diagnostics confirm the mechanism under the hood, i.e., unconstrained learners have redundant, weakly active tangent directions whereas symmetry imposition prunes auxiliary flat directions, maximizing the locally useful target volume per search direction.

Taken together, these results suggest that pruning NQS model parameters by imposing symmetry constraints via structural degeneracies regularizes the geometry of the learning dynamics. These methods therefore drive the model to obtain better gradient directions and furnish a highly scalable and highly accurate state for studying many-body physics.

\section{Theoretical Framework}

A wide variety of neural quantum state architectures have been proposed in the literature, including restricted Boltzmann machines, autoregressive models, convolutional networks, correlator-product states, and learned-coefficient parameterizations~\cite{carleo2017solving,hibat2020recurrent,barrett2022autoregressive,sharir2022neural,Reh2023}. While these architectures vary structurally, they address the identical underlying task: for each computational-basis configuration, learning the corresponding amplitude and phase for a quantum state. Rather than analyzing each machine learning architecture microscopically, we introduce a generalized framework under which our analysis holds universally.

\subsection{The spectrum of the NQS ansätze}

Let $\mathcal{X} = (\mathbb{C}^2)^{\otimes n}$ denote the $2^n$-dimensional complex Hilbert space of an $n$-qubit system, and let $\mathcal{B}_c = \{|\vec{z}\rangle \mid \vec{z} \in \{0,1\}^n\}$ be the standard computational basis representing classical spin configurations. Following the canonical neural quantum state (NQS) framework~\cite{lange2024architecturesapplicationsreviewneural, carleo2017solving}, a pure quantum state $|\psi\rangle \in \mathcal{X}$ can be compactly parameterized by expressing its wave function coefficients directly through a neural network or variational architecture. To isolate the modulation of the quantum probability distribution from the underlying phase structure, we adopt the canonical amplitude-phase representation~\cite{lange2024architecturesapplicationsreviewneural}:
\begin{equation}
    |\psi_{\Xx, \Yy}\rangle = \sum_{\vec{z}\in\{0,1\}^n} \sqrt{P_{\Xx}(\vec{z})} e^{-i \phi_{\Yy}(\vec{z})} |\vec{z}\rangle , \label{eq:gen_nqs_canonical}
\end{equation}
where $P_{\Xx}(\vec{z}) \ge 0$ represents the Born probability distribution satisfying $\sum_{\vec{z}} P_{\Xx}(\vec{z}) = 1$, and $\phi_{\Yy}(\vec{z}) \in \mathbb{R}$ denotes the relative phase field. The tunable parameter vectors are split into amplitude-network parameters $\Xx\in\mathbb R^d$ and phase-network parameters $\Yy\in\mathbb R^k$, and we assume that $d$ and $k$ grow at most polynomially with the system size.

To construct a flexible, physically interpretable framework that encompasses standard architectures found in the literature, we parameterize these fields using diagonal operator expansions acting on a reference state. Let $\mathcal{B} = \left\{\bigotimes_{j=1}^n \sigma_j^{a_j} \;\middle|\; a_j \in \{z, I\}\right\}$ denote the abelian group of diagonal Pauli strings. We define parameterized Hermitian generators $H(\Xx), G(\Yy) \in \operatorname{Herm}(\mathcal{X})$ restricted to this basis:
\begin{equation}
    H(\Xx) = \sum_{\alpha\in\mathcal{B}} c_\alpha(\Xx) P_\alpha, \quad G(\Yy) = \sum_{\alpha\in\mathcal{B}} d_\alpha(\Yy) P_\alpha,
\end{equation}
where $c_\alpha(\Xx), d_\alpha(\Yy) \in \mathbb{R}$ represent the output fields of the neural network architecture evaluated at parameters $\Xx$ and $\Yy$.

Since the elements of $\mathcal{B}$ are mutually commuting, diagonal operator strings, the computational basis states $|\vec{z}\rangle$ are simultaneous eigenstates of every $P_\alpha \in \mathcal{B}$, satisfying $P_\alpha |\vec{z}\rangle = p_\alpha(\vec{z})|\vec{z}\rangle$ with eigenvalues $p_\alpha(\vec{z}) \in \{-1, 1\}$. Consequently, the functions of these operators act diagonally on the computational basis, mapping the operator exponentials directly to scalar fields~\cite{odonnell2014analysis}.

The physical structure of the state $|\psi_{\Xx,\Yy}\rangle$ is defined by the factored ansatz \cite{sajjan2024polynomially,melko2019restricted}:
\begin{equation}
    |\psi_{\Xx,\Yy}\rangle = \mathcal{N}^{-1/2} e^{-\frac{\beta}{2}H(\Xx)} e^{-iG(\Yy)} \left(\bigotimes_{j=1}^n e^{-i\theta_j\sigma^x_j}\right) |\vec{0}\rangle,
\end{equation}
where $\beta \in \mathbb{R}$ physically represents a fixed inverse temperature scaling factor, and $|\vec{0}\rangle = |0\rangle^{\otimes n}$. The local $\sigma^x$ rotations seed a reference product-state superposition with baseline amplitudes $\alpha_{\vec{z}}(\vec{\theta}) \equiv \langle\vec{z}|\bigotimes_{j=1}^n e^{-i\theta_j\sigma^x_j}|\vec{0}\rangle$. The diagonal operators $e^{-\frac{\beta}{2}H(\Xx)}$ and $e^{-iG(\Yy)}$ modulate this reference structure without inducing basis mixing.

Projecting this ansatz onto the computational basis yields the real eigenvalue fields $\eta(\Xx,\vec{z}) \equiv \sum_{\alpha\in\mathcal{B}} c_{\alpha}(\Xx)p_{\alpha}(\vec{z})$ and $\gamma(\Yy,\vec{z}) \equiv \sum_{\alpha\in\mathcal{B}} d_{\alpha}(\Yy)p_{\alpha}(\vec{z})$. Consequently, the model is determined entirely by the coefficient vectors:
\begin{equation}
    \vec c(\Xx) = \left(c_{\alpha}(\Xx)\right)_{\alpha \in \mathcal B}, \quad \vec d(\Yy) = \left(d_{\alpha}(\Yy)\right)_{\alpha \in \mathcal B},
\end{equation}
which independently dictate the Born probabilities and relative phases via
\begin{align}
    P_{\Xx}(\vec{z}) = \frac{e^{-\beta\eta(\Xx,\vec{z})}|\alpha_{\vec{z}}(\vec{\theta})|^2}{\mathcal{N}}, \\ 
    \phi_{\Yy}(\vec{z}) = -\gamma(\Yy,\vec{z}) - \arg\alpha_{\vec{z}}(\vec{\theta}),
\end{align}
where $\mathcal{N} \equiv \sum_{\vec{z}} e^{-\beta\eta(\Xx,\vec{z})}|\alpha_{\vec{z}}(\vec{\theta})|^2$ is the partition function acting as the normalization constant.

\begin{figure*}[t]
    \centering
    \includegraphics[width=\textwidth]{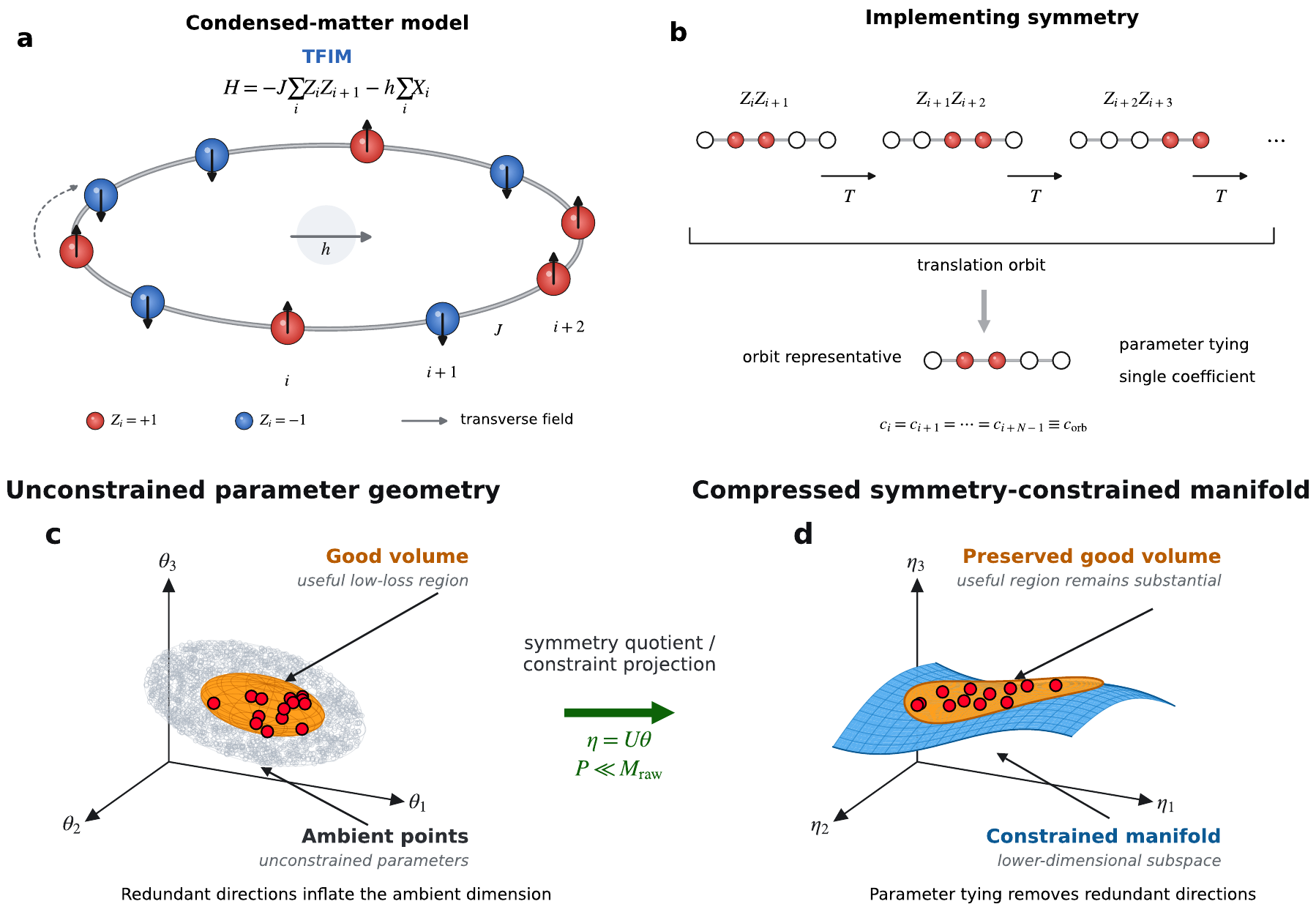}
    \caption{\textbf{Symmetry compresses the neural-quantum-state parameter manifold while preserving its useful region.} \textbf{a}, Periodic transverse-field Ising model (TFIM), $H=-J\sum_i Z_iZ_{i+1}-h\sum_iX_i$, with nearest-neighbour coupling $J$ and transverse field $h$. \textbf{b}, Lattice translations map each local interaction support to an orbit of symmetry-equivalent supports. Enforcing symmetric constraints by imposing translational invariance ties the corresponding coefficients, replacing the full set of translated parameters by a single orbit coefficient~\cite{choo2018, RBM_anyons_symm}. \textbf{c}, In the unconstrained parametrization, symmetry-equivalent descriptions generate redundant directions that enlarge the ambient parameter space; the green region denotes parameters satisfying a prescribed accuracy or energy criterion. \textbf{d}, Quotienting these redundant directions produces a lower-dimensional symmetry-constrained manifold. The image of the good region remains extended on the reduced manifold, illustrating that symmetry reduction can remove redundant degrees of freedom without collapsing the volume of accurate, trainable solutions.}
    \label{fig:symmetry-manifold-compression}
\end{figure*}

Given the universal expressivity of these fields over $\mathbb{Z}_2^n$~\cite{sajjan2024polynomially}, we variationally minimize $\langle f(\mathcal{H})\rangle_{\rho} = \mathbb{E}_{\vec{z}_1 \sim P(\vec{z}_1)} [f_{\mathrm{loc}}(\mathcal{H}, \vec{z}_1)]$ via Monte Carlo sampling~\cite{carleo2017solving,sajjan2024polynomially}.

For a $k$-local target Hamiltonian $\mathcal H$, the standard VMC local energy at a sampled configuration $\vec z_1$ is $E_{\mathrm{loc}}(\vec z_1)\equiv\sum_{\vec z_2}\langle\vec z_1|\mathcal H|\vec z_2\rangle\langle\vec z_2|\psi_{\Xx,\Yy}\rangle/\langle\vec z_1|\psi_{\Xx,\Yy}\rangle$~\cite{carleo2017solving,Reh2023}. For Pauli terms with support on at most $k$ sites, only configurations connected to $\vec z_1$ by flips on at most $k$ sites contribute to this sum. Nonetheless, unconstrained parameterization of $\vec c(\Xx)$ and $\vec d(\Yy)$ introduces redundant, weakly identifiable directions that generate flat landscapes and ill-condition the optimization~\cite{RBM_Geom_learning,PhysRevB.100.195125,dash2025efficiency,Reh2023}. To resolve this ill-conditioning, we systematically truncate the expansion basis $\mathcal{B}$. The following section details how imposing locality constraints on the diagonal generators enables direct eigenvalue engineering of the amplitude and phase fields, effectively restoring trainability.

\subsection{Eigenvalue engineering via symmetry constraints}
To regularize the variational parameterization, we first restrict the generators to a $k$-local Pauli basis and then impose symmetry constraints that identify symmetry-equivalent coefficients. 
While the locality restriction limits the operator basis, it does not by itself reduce the number of distinct eigenvalues.  Let $S_k = \{\vec{\alpha} \in \{0,1\}^n \mid \|\vec{\alpha}\|_1 \le k\}$ index the $k$-local operators, where $\|\cdot\|_1$ denotes the Hamming weight. The truncated generators are:
\begin{equation}
    H(\Xx) = \sum_{\vec{\alpha} \in S_k} c_{\vec{\alpha}}(\Xx) P_{\vec{\alpha}}, \quad G(\Yy) = \sum_{\vec{\alpha} \in S_k} d_{\vec{\alpha}}(\Yy) P_{\vec{\alpha}},
\end{equation}
where $P_{\vec{\alpha}} = \bigotimes_{j=1}^n (\sigma^z_j)^{\alpha_j}$. Limiting the distinct eigenvalues of $H(\Xx)$ and $G(\Yy)$ independently bounds the joint spectral degrees of freedom. By symmetry, it suffices to analyze the spectrum of $H(\Xx)$ mapped from its coefficient vector $\vec{c}(\Xx) \in \mathbb{R}^{M_k}$, where $M_k = |S_k| = \sum_{m=0}^k \binom{n}{m}$.

The spectrum of $H(\Xx)$ defines a linear map $\vec{\eta} = W\vec{c}$, where the design matrix $W \in \mathbb{R}^{2^n \times M_k}$ has elements $W_{\vec{z}, \vec{\alpha}} = (-1)^{\vec{\alpha} \cdot \vec{z}}$ for $\vec{z} \in \{0,1\}^n$. The columns of $W$ correspond to orthogonal characters of the abelian group $\mathbb{Z}_2^n$, forcing the matrix to satisfy $W^T W = 2^n \mathbb{I}_{M_k}$. Consequently, $\operatorname{rank}(W) = M_k$, ensuring the transformation from $k$-local coefficients to the eigenvalue spectrum is strictly injective.

To prescribe a degeneracy pattern, let $\{\mathcal C_j\}_{j=1}^{D}$ be a partition of the computational-basis configurations into $D$ collision classes, where $|\mathcal C_j|=m_j$ and $\sum_{j=1}^{D}m_j=2^n$. For each class $\mathcal C_j$, choose a representative configuration $\vec z^{(j,0)}\in\mathcal C_j$ and enumerate the remaining configurations as $\vec z^{(j,a)}$ for $a=1,\ldots,m_j-1$. Requiring every configuration in $\mathcal C_j$ to share the eigenvalue of its representative produces $\sum_{j=1}^{D}(m_j-1)=2^n-D$ homogeneous linear constraints on $\vec c$.

These constraints are collected in the configuration-space collision matrix $V_{\text{coll}}\in\mathbb R^{(2^n-D)\times M_k}$, whose entries are given by:
\begin{align}
    (V_{\text{coll}})_{(j,a),\vec\alpha} &= W_{\vec z^{(j,0)},\vec\alpha}-W_{\vec z^{(j,a)},\vec\alpha} \nonumber \\
    &=(-1)^{\vec\alpha\cdot\vec z^{(j,0)}}-(-1)^{\vec\alpha\cdot\vec z^{(j,a)}}, \label{eq:collision-matrix}
\end{align}
where $j=1,\ldots,D$ and $a=1,\ldots,m_j-1$. Since $\eta(\vec c,\vec z)=\sum_{\vec\alpha\in S_k}W_{\vec z,\vec\alpha}c_{\vec\alpha}$, each component of the constraint vector maps directly to an eigenvalue difference:
\begin{equation}
    (V_{\text{coll}}\vec c)_{(j,a)}=\eta(\vec c,\vec z^{(j,0)})-\eta(\vec c,\vec z^{(j,a)}).
\end{equation}
Therefore, $V_{\text{coll}}\vec c=\vec 0$ is equivalent to requiring that all configurations within each class $\mathcal C_j$ share a common eigenvalue. The resulting spectrum contains at most $D$ distinct eigenvalues because configurations assigned to different classes may still acquire the same eigenvalue.

To construct this allowable parameter subspace without evaluating $V_{\text{coll}}$ via an explicit, computationally intractable $\mathcal{O}(2^n)$ state-space enumeration, we exploit the underlying discrete spatial and internal symmetry groups $\mathcal{G}$ of the lattice system~\cite{choo2018,RBM_anyons_symm,sauvage2024building}. A spatial symmetry operation $g \in \mathcal{G}$ acts as a site permutation, inducing a linear representation on the $k$-local coefficient basis. The corresponding permutation matrix $P_g \equiv \mathcal{P}(g) \in \mathbb{R}^{M_k \times M_k}$ acts directly on the coefficients via $(P_g \vec{c})_{\alpha} = c_{g^{-1}(\alpha)}$. Physical invariance of the generator requires $(\mathbb{I}_{M_k} - P_g)\vec{c} = \vec{0}$ for all $g \in \mathcal{G}$. The global constraint matrix $V$ can be constructed systematically by serially evaluating a minimal generating set $\{g_1, g_2, \dots, g_m\} \subseteq \mathcal{G}$ alongside any internal parity operators:

\begin{enumerate}
    \item \emph{Lattice Translations ($C_n$):} For a 1D periodic chain generated by the single-site cyclic translation $\mathsf{T}$, the action on the basis is $P_{\mathsf{T}\vec{\alpha}} = \mathsf{T}P_{\vec{\alpha}}\mathsf{T}^{-1}$, forcing $c_{\vec{\alpha}} = c_{\mathsf{T}\vec{\alpha}}$. This yields the constraint block $V_T = \mathbb{I}_{M_k} - P_T$, which compresses the parameter space via Burnside's lemma to $\mathcal{G}$-orbits and maps the $2^n$ computational states onto $D_{\lambda}(C_n) = \frac{1}{n}\sum_{r=0}^{n-1} 2^{\gcd(n,r)}$ degenerate eigenvalue necklace orbits~\cite{burnside1911,redfield1927,polya1937,harary_palmer1973}.
    
    \item \emph{Full Space Groups ($D_n$):} Adding spatial reflection $\mathsf{R}$ (where $\mathsf{R}^2=\mathbb{I}$ and $\mathsf{R}\mathsf{T}\mathsf{R}=\mathsf{T}^{-1}$) accounts for orientation reversal, mapping crystal momentum $q \to -q$ and folding the Brillouin zone. The space-group constraint is captured by introducing the block $V_R = \mathbb{I}_{M_k} - P_R$, enforcing $c_{\vec{\alpha}} = c_{\mathsf{R}\vec{\alpha}}$. This collapses the eigenvalues into $D_{\lambda}(D_n) = \frac{1}{2n} \left( \sum_{r=0}^{n-1} 2^{\gcd(n,r)} + \sum_{r=0}^{n-1} 2^{c(\mathsf{R}\mathsf{T}^r)} \right)$ binary bracelet classes, where $c(\mathsf{R}\mathsf{T}^r)$ is the cycle count of the reflected translation~\cite{polya1937,harary_palmer1973}.
    
    \item \emph{Point-Group Symmetries ($\mathcal{P}$):} For higher-dimensional lattices, operations $R \in \mathcal{P}$ leave a designated origin fixed while permuting site indices (e.g., $D_4$ rotations and reflections on a 2D square lattice). Generating blocks of the form $V_R = \mathbb{I}_{M_k} - P_R$ for a minimal generating set of $\mathcal{P}$ eliminates artificial bond anisotropy (forcing $c_{\hat{x}} = c_{\hat{y}}$) and restricts the spectral degrees of freedom to $D_{\lambda}(\mathcal{P}) = \frac{1}{|\mathcal{P}|} \sum_{R \in \mathcal{P}} 2^{c(R)}$ orbit classes~\cite{polya1937,harary_palmer1973}.
    
    \item \emph{Global Bitflip Parity ($\mathsf{B}$):} Unlike spatial mappings, the internal spin-flip operator $\mathsf{B} = \prod_{i=1}^n X_i$ acts as an algebraic selection rule. Because $\mathsf{B}P_{\vec{\alpha}}\mathsf{B}^{-1} = (-1)^{|\vec{\alpha}|}P_{\vec{\alpha}}$, invariance under $\mathsf{B}$ forces $c_{\vec{\alpha}} = 0$ for all odd Hamming weights $|\vec{\alpha}|$. This eliminates the entire odd-support sector via a diagonal projection block $V_B = \operatorname{diag}(\vec{v})$, where $v_{\alpha} = 1$ if $|\vec{\alpha}| \equiv 1 \pmod 2$ and $0$ otherwise. Because $\mathsf{B}$ has no fixed points on $\mathcal{B}_c$, it uniformly splits the Hilbert space into $2^{n-1}$ degenerate doublets $\{\vec{z}, \vec{z}\oplus\vec{1}\}$~\cite{choo2018,RBM_anyons_symm}.
\end{enumerate}

Vertically stacking these operational blocks yields the complete global constraint matrix
\begin{equation}
    V = \begin{bmatrix} \mathbb{I}_{M_k} - P_{g_1} \\ \vdots \\ \mathbb{I}_{M_k} - P_{g_m} \\ V_B \end{bmatrix} \in \R^{(m+1)M_k \times M_k}.
\end{equation}

The subspace of physically allowable, symmetry-invariant coefficients is precisely defined by the null space $\text{Null}(V) \subseteq \mathbb{R}^{M_k}$. Let $U_V \in \mathbb{R}^{M_k \times d_{\text{null}}}$ be an isometric matrix whose columns form an orthonormal basis spanning $\text{Null}(V)$, where $d_{\text{null}} = \dim(\text{Null}(V))$. Reparameterizing the coefficient vector via $\vec{c} = U_V \vec{\xi}$ for an unconstrained coordinate vector $\vec{\xi} \in \mathbb{R}^{d_{\text{null}}}$ directly simplifies the spectral mapping to $\vec{\eta} = W_{\text{red}} \vec{\xi}$. Here, the reduced design matrix is explicitly given by
\begin{equation}
    W_{\text{red}} = W U_V ,
\end{equation}
which maps the unconstrained variational coordinates directly onto the allowable eigenvalue subspace. This construction restricts the optimization trajectory strictly to the symmetry-preserving tangent space, mathematically excising redundant, ill-conditioned directions and regularizing the geometric capacity of the network~\cite{RBM_Geom_learning,dash2025efficiency,sauvage2024building}. Figure~\ref{fig:symmetry-manifold-compression} summarizes the central mechanism wherein spatial symmetry identifies lattice-shifted interaction supports and ties their coefficients, thereby quotienting symmetry-redundant directions from the variational parameter space while retaining an extended region of accurate solutions on the compressed manifold.

\section{Expressibility and Trainability of an NQS}
\label{sec:trainability_nqs}

Beyond merely reducing the number of free variables, incorporating physical symmetries directly into the ansatz can regularize the optimization landscape by restricting the variational search to symmetry-compatible directions~\cite{choo2018,sauvage2024building}. Our central contribution is a target-aware geometric metric that distinguishes raw parameter count, physically accessible variational capacity, and the fraction of that capacity concentrated near accurate solutions. The trainability of a variational quantum ansatz is governed by the geometry and conditioning of its parameter manifold~\cite{RBM_Geom_learning,haug2021capacity,dash2025efficiency}. The Jacobian defines the accessible tangent space and the pullback quantum geometric tensor, whose real part yields the Fubini--Study metric and is proportional, up to convention, to the pure-state Quantum Fisher Information Matrix~\cite{provost1980riemannian,stokes2020quantum}. The objective Hessian, in turn, determines the local curvature of the energy landscape. Here, we explicitly derive these geometric quantities for the generalized NQS model. By comparing the unconstrained landscape with the symmetry-projected subspace, we demonstrate that hard-wiring the physical invariants excises symmetry-forbidden tangent directions and their associated structural zero modes, thereby regularizing the physically relevant optimization landscape. Proofs are deferred to Appendix~\ref{app:jacobian_proofs}.

\subsection{Tangent space geometry of the variational manifold- Local notion of expressibility}

Using the diagonal operators $H(\Xx)$ and $G(\Yy)$ defined above, their computational-basis eigenvalues are $\eta(\Xx,\vec z)=\sum_{\alpha\in\mathcal B}c_\alpha(\Xx)p_\alpha(\vec z)$ and $\gamma(\Yy,\vec z)=\sum_{\alpha\in\mathcal B}d_\alpha(\Yy)p_\alpha(\vec z)$. The resulting NQS amplitude is
\begin{equation}
\psi_{\Xx,\Yy}(\vec z)=\frac{e^{-\frac{\beta}{2}\eta(\Xx,\vec z)}e^{-i\gamma(\Yy,\vec z)}\alpha_{\vec z}(\vec\theta)}{\sqrt{\mathcal N(\Xx)}},
\end{equation}
with $\mathcal N(\Xx)=\sum_{\vec z}e^{-\beta\eta(\Xx,\vec z)}|\alpha_{\vec z}(\vec\theta)|^2$. Consequently, the Jacobian may be defined as follows.

\begin{lemma}[Unconstrained Jacobian]
\label{lem:uncon_jacobian}
For $\vec\theta_{\rm un}=(\vec c(\Xx),\vec d(\Yy))\in\mathbb R^{2|\mathcal B|}$, the Jacobian $J_{\rm uncon}=(J^{(c)}, J^{(d)})\in\mathbb C^{2^n\times 2|\mathcal B|}$ has matrix elements
\begin{align}
\label{eq:uncon-jacobian}
J^{(c)}(\vec z,\beta)=\frac{\partial \psi_{{\Xx,\Yy}}(\vec z)}{\partial c_\beta(\Xx)}&=-\frac{\beta}{2}\psi_{\Xx,\Yy}(\vec z)\left[p_\beta(\vec z)-\left\langle p_\beta\right\rangle_\psi\right],\nonumber\\
J^{(d)}(\vec z,\beta)=\frac{\partial \psi_{{\Xx,\Yy}}(\vec z)}{\partial d_\beta(\Yy)}&=-ip_\beta(\vec z)\psi_{\Xx,\Yy}(\vec z).
\end{align}
where $\langle\cdot\rangle_\psi$ is the expectation over the Born distribution $|\psi_{\Xx,\Yy}(\vec z)|^2$.
\end{lemma}

To enforce physical symmetries, we restrict $\vec{\tilde c}(\Xx,\Yy)=\vec c(\Xx)+i\vec d(\Yy)\in\mathbb C^{|\mathcal B|}$ to the null space of a linear constraint $V\vec{\tilde c}(\Xx,\Yy)=\vec 0$. Parameterizing this via a null-space basis $U\in\mathbb C^{|\mathcal B|\times d_\xi}$ yields $\vec{\tilde c}(\Xx,\Yy)=U\vec\xi$ for $\vec\xi\in\mathbb C^{d_\xi}$. Writing $\xi_a=\xi_a^R+i\xi_a^I$, $U_{\alpha a}=U_{\alpha a}^R+iU_{\alpha a}^I$, and defining the projected coefficients $A_{\vec z,a}=A_{\vec z,a}^R+iA_{\vec z,a}^I\equiv\sum_{\alpha\in\mathcal B}p_\alpha(\vec z)U_{\alpha a}$, the constrained eigenvalue fields become
\begin{align}
\eta(\vec\xi,\vec z)&=\sum_{a=1}^{d_\xi}\left(A_{\vec z,a}^R\xi_a^R-A_{\vec z,a}^I\xi_a^I\right),\\
\gamma(\vec\xi,\vec z)&=\sum_{a=1}^{d_\xi}\left(A_{\vec z,a}^R\xi_a^I+A_{\vec z,a}^I\xi_a^R\right).
\end{align}

{Subsequently, the Jacobian for the spectral-constrained models is given as follows.}
\begin{lemma}[Constrained Jacobian]
\label{lem:con_jacobian}
For $\vec\theta_{\rm con}=(\vec\xi^R,\vec\xi^I)\in\mathbb R^{2d_\xi}$, the Jacobian $J_{\rm con}=(J^{(R)}\;J^{(I)})\in\mathbb C^{2^n\times 2d_\xi}$ has matrix elements
\begin{align}
\label{eq}
J^{(R)}{\vec z,a}=\frac{\partial\psi_{\vec\xi}(\vec z)}{\partial\xi_a^R}&=-\psi_{\vec\xi}(\vec z)\left[\frac{\beta}{2}\left(A_{\vec z,a}^R-\langle A_a^R\rangle_\psi\right)+iA_{\vec z,a}^I\right],\nonumber\\
J^{(I)}{\vec z,a}=\frac{\partial\psi_{\vec\xi}(\vec z)}{\partial\xi_a^I}&=-\psi_{\vec\xi}(\vec z)\left[-\frac{\beta}{2}\left(A_{\vec z,a}^I-\langle A_a^I\rangle_\psi\right)+iA_{\vec z,a}^R\right].
\end{align}

\end{lemma}

{
Mathematically, the columns of the Jacobian $J$ span the tangent space $T_{\psi}\mathcal{M}$ of the parameter manifold, representing the directional derivatives of the state representation. To best appreciate this point, consider an infinitesimal variation of the real parameter vector \(d\vec{\theta}\), where $(\vec{\theta}=(\vec{c},\vec{d})$ for the unconstrained ansatz or $\vec{\theta}=(\vec{\xi}^{\,R},\vec{\xi}^{\,I})$ for the symmetry-constrained ansatz). To first order, the corresponding variation of the NQS wavefunction is $d|\psi(\vec{\theta})\rangle
=J (\vec{\theta})\,d\vec{\theta}$ where $J(\vec{\theta})$ can either be $J_{\rm{uncon}}$ or $J_{\rm{con}}$ depending on whether we are considering the unconstrained or the constrained case (see Lemma \ref{lem:uncon_jacobian} or Lemma \ref{lem:con_jacobian}). This infinitesimal displacement can contain a component parallel to the current state, such as a change in the overall global phase, which does not correspond to a physically distinguishable change of the quantum state. To remove such gauge-like contributions, we project the displacement onto the subspace orthogonal to $|\psi(\vec{\theta})\rangle$:
\begin{align}
d|\psi_{\perp}(\vec{\theta})\rangle &=\Pi_{\perp}J(\vec{\theta})\,d\vec{\theta}\\
&=\left(\mathbb{I}-|\psi(\vec{\theta})\rangle\langle\psi(\vec{\theta})|\right)J(\vec{\theta})d\vec{\theta}.
\end{align}
The squared norm of this physically distinguishable displacement is therefore
\begin{eqnarray}\label{eq:sq_norm_S_metric}
\left|d|\psi_{\perp}(\vec{\theta})\rangle\right|^2=d\vec{\theta}^{\,T}J^\dagger(\vec{\theta})\Pi_{\perp}J(\vec{\theta})d\vec{\theta}.
\end{eqnarray}
For real parameter variations, this induces the pullback Fubini--Study metric on the variational parameter manifold~\cite{provost1980riemannian,stokes2020quantum,dash2025efficiency},
\begin{eqnarray}\label{eq:defn_S_metric}
S(\vec{\theta})=\operatorname{Re}\left[J^\dagger(\vec{\theta})\Pi_{\perp}J(\vec{\theta})\right].
\end{eqnarray}
The metric $S(\vec{\theta})$ therefore quantifies how sensitively the physical quantum state changes under infinitesimal variations of the variational parameters. 

Rank deficiency of $\Pi_{\perp}J$, and hence of $S(\vec{\theta})$, identifies null parameter directions that produce no first-order displacement of the physical state. Exact null directions correspond to locally redundant or gauge-like parameter variations that produce no physically distinguishable first-order state displacement, while near-null directions signal weak local identifiability~\cite{haug2021capacity,dash2025efficiency}. Consequently, an important quantity characterizing the physical local expressive dimension is the rank~\cite{haug2021capacity,dash2025efficiency}
\begin{eqnarray}
r_{\mathrm{uncon/con}}(\vec{\theta}_{\mathrm{un/con}})=\rank S_{\mathrm{uncon/con}}(\vec{\theta}_{\mathrm{un/con}}),
\end{eqnarray}
which defines the number of linearly independent and physically distinguishable tangent directions accessible at $|\psi(\vec{\theta})\rangle$ in either the unconstrained or constrained case. It necessarily satisfies

\begin{eqnarray}
r_{\mathrm{uncon/con}}(\vec{\theta})
\leq
\min\!\left\{
|\vec{\theta}|,
2^{n+1}-2
\right\},
\end{eqnarray}
where $2^{n+1}-2$ is the real dimension of the projective Hilbert space of normalized $n$-qubit pure states. This tangent-space geometry provides the differential-geometric foundation for the energy-landscape curvature further analyzed in the following sections.}
%

\subsection{Geometry of the optimization landscape- curvature of the objective function}

{
The previous section established the geometry of the variational manifold of the defined NQS through the Jacobian and its relation to the Fubini-Study metric. The Jacobian identifies
physically accessible tangent directions of the ansatz and characterizes its local expressive capacity. However, expressibility alone does not determine how easily the ansatz can be optimized, because highly expressive models may still exhibit poorly scaled gradients or unfavorable parameter-space geometry~\cite{holmes2022connecting,haug2021capacity}. This information is thus encoded within the geometry of the objective function and hence demands a target-aware analysis. Understanding the local curvature of the objective function can provide a glimpse of the near-flat directions and high-curvature valleys drive landscape ill-conditioning, which severely throttles optimization step sizes. In variational quantum algorithms, this ill-conditioning frequently arises from over-parameterization. Once the number of coordinates exceeds the maximal dimension of the physically accessible state-space directions, the ranks of the QFIM and objective Hessian saturate, and additional parameters become locally redundant~\cite{larocca2023overparametrization,haug2021capacity}. This rank saturation can improve the global structure of the optimization landscape while simultaneously rendering raw coordinate-space metric or Hessian inversions singular or ill-conditioned. The relevant issue for local optimization is therefore not parameter count alone, but the resolved spectrum on the physically active tangent subspace. In the present construction, symmetry projection restricts the parameterization to the constraint-satisfying subspace and removes symmetry-forbidden directions before optimization.}

\begin{lemma}[Energy gradient and Hessian]
\label{lem:hessian}
Let $E(\vec\theta) = \langle\psi(\vec\theta)|\mathcal H|\psi(\vec\theta)\rangle$ for Hermitian $\mathcal H$, with $\vec\theta \in \{\vec\theta_{\rm un},\vec\theta_{\rm con}\}$ denoting the unconstrained and constrained parameters. The landscape derivatives are:
\begin{align}
\frac{\partial E}{\partial\theta_\mu} &= 2\,\mathrm{Re}\big[\langle\partial_\mu\psi|\mathcal H|\psi\rangle\big], \\
\frac{\partial^2 E}{\partial\theta_\mu\partial\theta_\nu} &= 2\,\mathrm{Re}\big[\langle\partial_\mu\psi|\mathcal H|\partial_\nu\psi\rangle + \langle\partial_\mu\partial_\nu\psi|\mathcal H|\psi\rangle\big]\label{eq:grad_hess}.
\end{align}
In matrix form, the Hessian is
\begin{align}
\label{eq:hessian_structure}
\nabla^2_{\vec\theta}E &= 2\,\mathrm{Re}\big[J^\dagger\mathcal H J + R\big],\\
R_{\mu\nu} &= \langle\partial_\mu\partial_\nu\psi|\mathcal H|\psi\rangle,
\end{align}
where $J\in\{J_{\rm uncon}, J_{\rm con}\}$ denotes the unconstrained and the constrained Jacobian respectively.
\end{lemma}

Equation~\eqref{eq:hessian_structure} partitions the curvature into a Jacobian-driven metric term $J^\dagger\mathcal H J$ and a second-derivative correction $R$. The first term has the Jacobian Gram structure of a pullback metric, but with the Hamiltonian inserted as an energy-dependent weight. The term $R$ captures the intrinsic nonlinear curvature of the ansatz state space and is 
responsible for not making the full Hessian dependent on the Jacobian alone as defined in previous section.

At a stationary point $\vec\theta_*$, let $M_* \equiv \nabla_{\vec\theta}^2 E(\vec\theta_*)$ denote the energy Hessian. Note that the raw parameter directions need not correspond to distinguishable changes of the quantum state so the curvature must be interpreted on the physical tangent space defined by the Fubini-Study metric $S_* \equiv S(\vec\theta_*)$. 
{
Equivalently, we observe that the metric tensor of the physically distinguishable state space generated by the ansatz $|\psi(\vec{\theta})\rangle$ is not the usual identity but $S(\vec{\theta})$ as defined in Eqs.~\eqref{eq:sq_norm_S_metric} and \eqref{eq:defn_S_metric}. This essentially means that the state space is not equally sensitive to all coordinates. Thus to redefine a new set of scale-invariant Euclidean coordinates one can do $d \vec{\phi}^* = S_*^{1/2} d\vec{\theta}^*$. Equivalently, the metric tensor of the physically distinguishable state space generated by the ansatz $|\psi(\vec{\theta})\rangle$ is not the coordinate-space identity but the Fubini--Study metric $S(\vec{\theta})$ defined in Eqs.~\eqref{eq:sq_norm_S_metric} and~\eqref{eq:defn_S_metric}. Consequently, the physical state is not equally sensitive to variations along all parameter directions. At the stationary point $\vec{\theta}_*$, the infinitesimal transformation $d\vec{\phi}_*=S_*^{1/2}d\vec{\theta}_*$ introduces locally metric-normalized coordinates on the physically active tangent subspace. We therefore define the Fubini--Study-normalized physical curvature operator
\begin{equation}\label{eq:normalized_curvature}
K_* \equiv \widetilde{S}_*^{-1/2}M_*\widetilde{S}_*^{-1/2},
\end{equation}
where $\widetilde{S}_*^{-1/2}$ is the square root of the Moore--Penrose pseudoinverse restricted to the physically active tangent subspace. If the eigenvalues of $K_*$ are $\{\mu_i\}$, we define the positive-curvature rank as the cardinality of the resolved curvature spectrum:
\begin{equation}
\label{eq:r_plus_definition}
r_+ \equiv \left| \{ i : \mu_i > \tau_{\mathrm H} \} \right|,
\end{equation}
where $\tau_{\mathrm H} > 0$ is a numerical threshold used to distinguish genuinely confining directions from null or numerically unresolved modes. The associated positive pseudodeterminant is therefore given by
\begin{equation}
\label{eq:pos_det_K}
\det\nolimits_+ K_* \equiv \prod_{\mu_i > \tau_{\mathrm H}} \mu_i.
\end{equation}}
The quantity $r_+$ counts the physically distinguishable directions along which the energy rises to quadratic order around $\vec\theta_*$. It is target dependent and differs from the physical tangent rank $r_{\rm uncon/con}$. The rank $r_{\rm uncon/con}$ counts every locally accessible physical motion, while $r_+$ counts only the motions confined by the rise of the target energy. Thus the resulting hierarchy is
\begin{equation}
\label{eq:dimension_hierarchy}
r_+\le r_{\rm uncon/con}\le |\vec{\theta}|.
\end{equation}
This hierarchy separates target-confined capacity, physical expressive capacity, and raw coordinate capacity.

Without constraints, the unprojected Hessian $\nabla^2_{\vec\theta_{\rm un}}E_{\rm un}$ spans the full $\mathbb R^{2M_k\times 2M_k}$ space. The unconstrained ansatz can express symmetry-breaking parameter variations that leave the target physical state largely unaffected. Consequently, this unconstrained Hessian is heavily populated by zero or near-zero eigenvalues. 

At an unconstrained stationary point, the raw Hessian acts on the full coordinate space, including symmetry-breaking, gauge-like, and weakly identifiable directions. Such redundant or weakly identifiable directions can generate exact or near-zero curvature and thereby obscure the spectrum governing motion within the physically relevant basin~\cite{haug2021capacity,dash2025efficiency,larocca2023overparametrization}. Natural-gradient methods precondition the gradient using the QFIM or the real quantum geometric tensor, typically through a regularized linear solve, whereas Newton methods use the objective Hessian~\cite{stokes2020quantum,RBM_Geom_learning}. In both cases, the unconstrained structural rank-deficiency causes the relevant conditioning of the matrices to become poorly conditioned.

Symmetry projection dynamically alters this curvature landscape by restricting the optimization to the null space of the symmetry constraints. By mapping the parameterization directly onto the constraint-satisfying null space, we systematically excise these unphysical flat directions prior to optimization, compressing the landscape curvature to $\nabla^2_{\vec\theta_{\rm con}}E_{\rm con} \in \mathbb R^{2d_\xi\times 2d_\xi}$. This analytic deflation of structural zero-modes means the residual conditioning of the active physical space is characterized by the Fubini-Study-normalized condition number
\begin{equation}
\kappa_+ = \frac{\mu_{\max}^+}{\mu_{\min}^+},
\end{equation}
where $\mu_{\max}^+$ and $\mu_{\min}^+$ are the extreme resolved positive eigenvalues of $K_*$. Further, symmetry projection removes exact symmetry-forbidden modes and may improve the conditioning of the residual physical curvature spectrum. This reduction can lower the numerical regularization required by metric-aware or curvature-aware optimizers. Dimensional reduction alone does not guarantee a monotonic improvement of $\kappa_+$ because the outcome depends on how the retained subspace aligns with the target landscape.

\subsection{Solution space concentration and a local notion of trainability}

From the preceding subsections, we have established that raw parameter-space volumes are coordinate-dependent and susceptible to inflation by gauge redundancies, parameter rescalings, and locally redundant coordinates~\cite{haug2021capacity,dash2025efficiency,larocca2023overparametrization}. These artifacts can obscure the intrinsic geometry of the physical state manifold. To assess the geometric concentration of high-quality configurations, we measure their local physical volume relative to the total physical volume of the accessible state manifold. Quantum information geometry supplies a gauge-invariant metric on this manifold naturally~\cite{provost1980riemannian,stokes2020quantum}.

The Fubini--Study metric measures the local distinguishability of neighboring pure states and is proportional, up to convention, to the pure-state Quantum Fisher Information Matrix~\cite{provost1980riemannian,braunstein1994statistical,stokes2020quantum}. It is defined as
\begin{equation}\label{eq:fs_metric}
    S_{\mu\nu}=\operatorname{Re}\!\left[\langle\partial_\mu\psi|\Pi_\perp|\partial_\nu\psi\rangle\right]=\operatorname{Re}\!\left[J^\dagger\Pi_\perp J\right]_{\mu\nu},
\end{equation}
where $\Pi_\perp=\mathbb I-|\psi\rangle\langle\psi|$. For any real parameter-space tangent vector $\vec v$, the quadratic form $\vec v^{\,T}S\vec v=\|\Pi_\perp J\vec v\|^2$ vanishes precisely when $\vec v$ generates no physically distinguishable first-order displacement of the quantum state. Such kernel directions correspond to locally redundant or gauge-like parameter variations~\cite{haug2021capacity,dash2025efficiency}.

Let $\vec\theta\sim\vec\theta'$ when $|\psi(\vec\theta')\rangle=e^{i\phi}|\psi(\vec\theta)\rangle$ for some $\phi\in\mathbb R$, and let $\bar\Theta=\Theta/\!\sim$ denote the reduced parameter domain obtained by identifying physically equivalent representations of the same projective quantum state~\cite{provost1980riemannian,dash2025efficiency}. On each constant-rank chart of $\bar\Theta$, let $\vec\eta$ denote physical coordinates and let $S_{\rm red}(\vec\eta)$ be the corresponding nonsingular Fubini\textendash Study metric.

The total reachable physical volume of the ansatz manifold is given by:
\begin{equation}
\label{eq:total_fs_volume}
\mathrm{Vol}_{\rm FS}(\mathcal M) = \int_{\bar\Theta} \sqrt{\det\nolimits' S_{\rm red}(\vec\eta)}\,d\vec\eta.
\end{equation}
 Defining the domain in this manner prevents redundant coordinate covers from artificially multiplying the intrinsic physical volume~\cite{provost1980riemannian,stokes2020quantum}. Finiteness additionally requires the reduced reachable manifold to have finite Fubini--Study volume.

\begin{figure}[tb]
    \centering
    \includegraphics[width=\linewidth]{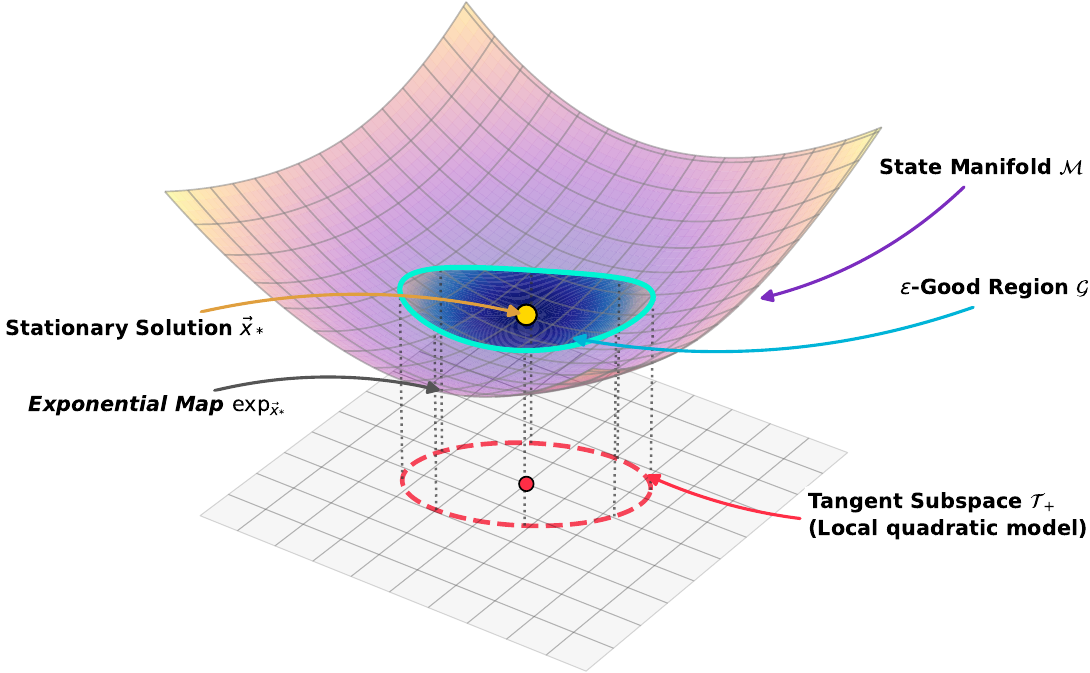}
    \caption{\textbf{Parameter landscape and local solution density.} 
    The curved upper surface represents the physical state manifold $\mathcal{M}$, hosting the $\epsilon$-good parameter basin $\mathcal{G}_\epsilon(\vec{x}_*)$ (cyan region) centered around the stationary solution $\vec{x}_*$ (gold marker). 
    The flat plane below represents the tangent subspace $\mathcal{T}_+$, which acts as the local quadratic model of the landscape. 
    The Fubini-Study exponential map $\exp_{\vec{x}_*}$ (represented by the vertical dashed projection lines) maps the boundary of the quadratic basin in $\mathcal{T}_+$ (dashed pink ellipse) onto the physical manifold, defining the local volume $V_{\text{good},+}^{\text{FS}}(\vec{x}_*; \epsilon)$ derived in Theorem~\ref{thm:good_volume}. 
    This coordinate-free framework quotients out exact parameter redundancies and unconfined gauge directions, which would otherwise render a Hessian-based full-dimensional basin unbounded, and thereby defines the target-aware useful-expressibility metric $f_\epsilon(\vec{x}_*)$ in Definition~\ref{def:trainability_metric}.}
    \label{fig:differential_geometry_trainability}
\end{figure}

To formalize this distinction, consider an enlarged ansatz whose physical manifold factorizes locally as $\mathcal M_B=\mathcal M_A\times\mathcal M_{\rm extra}$. Assume that the Fubini--Study measure and the acceptable region factorize as
\begin{align}\label{eq:product_factorization}
dV_{{\rm FS},B}&=dV_{{\rm FS},A}\,dV_{{\rm FS},{\rm extra}},\\
\mathcal G_{\epsilon,B}&=\mathcal G_{\epsilon,A}\times\mathcal R_{\rm extra}.
\end{align}
Defining
\begin{equation}
\label{eq:alpha_extra}
\alpha\equiv\frac{\operatorname{Vol}_{\rm FS}(\mathcal R_{\rm extra})}{\operatorname{Vol}_{\rm FS}(\mathcal M_{\rm extra})},
\end{equation}
the fractional volumes satisfy
\begin{equation}
\label{eq:product_fraction}
f_{\epsilon,B} = \alpha f_{\epsilon,A}.
\end{equation}
When the added directions are predominantly irrelevant to the target, $\alpha\ll1$ and $f_{\epsilon,B}\ll f_{\epsilon,A}$. Enlarging the reachable manifold does not improve useful expressibility unless the added physical volume contains a proportionate volume of acceptable states. Absolute basin width measures local robustness, while normalized basin volume measures useful expressibility.

We formalize this conceptual framework with the following definition:

\begin{definition}[Target-aware Useful Expressibility]
\label{def:trainability_metric}
Let $\vec x_*=[\psi(\vec\theta_*)]\in\mathcal M_\Theta$ be a stationary physical state. Let $\mathcal G_\epsilon(\vec x_*)$ denote the connected component containing $\vec x_*$ of the physical sublevel set whose energy lies within $\epsilon>0$ of $E(\vec x_*)$. The target-aware useful expressibility is
\begin{equation}
\label{eq:trainability_metric}
f_\epsilon(\vec x_*)\equiv\frac{\operatorname{Vol}_{\rm FS}\!\left[\mathcal G_\epsilon(\vec x_*)\right]}{\operatorname{Vol}_{\rm FS}(\mathcal M_\Theta)}.
\end{equation}
The numerator and denominator are measured on the same physical state manifold with the same intrinsic dimension.
\end{definition}

This dimensionless ratio $f_\epsilon(\vec\theta_*)$ measures the efficiency with which the ansatz allocates its physically distinguishable capacity to the target neighborhood. It is target-dependent through the energy landscape and tolerance $\epsilon$, and ansatz-dependent through both the local curvature and the total reachable manifold. The quantity $f_\epsilon$ is a geometric solution-density proxy and is not, by itself, the success probability of a particular optimization algorithm. Actual convergence also depends on the initialization distribution, gradient statistics, basin connectivity, intervening barriers, sampling noise, and update dynamics~\cite{RBM_Geom_learning,holmes2022connecting,mcclean2018barren}. Its predictive relation to optimization success must therefore be assessed empirically.

A trainable minimum must be evaluated through both the number and the stiffness of its confining directions. The positive-curvature rank $r_+$ specifies the intrinsic dimension of the local basin, while the spectrum of the physical curvature operator determines its width within that subspace. A rank-only diagnostic cannot distinguish broad and narrow minima of equal intrinsic dimension, whereas an ordinary determinant vanishes in the presence of exact null modes and, by itself, omits the dimension-dependent scaling with the tolerance. The pair $(r_+, \det_+ K_*)$ therefore provides the minimal spectral information required to quantify the local volume of acceptable solutions.

The Hessian determines the leading local geometry of the good region only within the quadratically confined physical directions. We now formalize this local contribution. Full derivations are provided in Appendix~\ref{app:jacobian_proofs}.

To evaluate the local volume without contamination from coordinate redundancies, we use the physical curvature operator $K_*$ defined in Eq.~\eqref{eq:normalized_curvature}. Let $\mathcal I_+=\{i\mid\mu_i>\tau_{\rm H}\}$ denote the resolved positive-curvature modes, with $r_+=|\mathcal I_+|$ and $\det_+K_*=\prod_{i\in\mathcal I_+}\mu_i$. Let $\mathcal T_+\subseteq T_{\vec x_*}\mathcal M_\Theta$ be the $r_+$-dimensional tangent subspace spanned by the corresponding eigenvectors, and let $\exp_{\vec x_*}$ denote the Fubini--Study exponential map.

\begin{theorem}[Local $\epsilon$-good volume at a stationary point]
\label{thm:good_volume}
Let $\vec x_*\in\mathcal M_\Theta$ be a stationary physical state that is a local minimum on the physically active tangent space and has no resolved negative-curvature modes. Let $\mathcal T_+\subseteq T_{\vec x_*}\mathcal M_\Theta$ be the $r_+$-dimensional physical tangent subspace spanned by the resolved positive-curvature eigenvectors of $K_*$. The Fubini--Study volume of the local $\epsilon$-good region restricted to $\mathcal T_+$ satisfies
\begin{align}
V_{{\rm good},+}^{\rm FS}(\vec x_*;\epsilon)&\equiv\operatorname{Vol}_{\rm FS}\!\left[\mathcal G_\epsilon(\vec x_*)\cap\exp_{\vec x_*}(\mathcal T_+)\right],\\
V_{{\rm good},+}^{\rm FS}(\vec x_*;\epsilon)&=\frac{\pi^{r_+/2}}{\Gamma(r_+/2+1)}\frac{(2\epsilon)^{r_+/2}}{\sqrt{\det\nolimits_+K_*}}\left[1+o(1)\right]
\end{align}
as $\epsilon\to0^+$.
\end{theorem}

Theorem~\ref{thm:good_volume} strips away coordinate artifacts to expose the true physical capacity of the local minimum. By restricting integration to the exponential map of the positive tangent subspace $\mathcal T_+$, the formulation systematically quotients out exact parameter redundancies and unconfined flat directions that would otherwise cause a naive full-dimensional volume formula to diverge. The resulting expression shows that the tolerance scaling is governed by the exponent $r_+/2$, while the aggregate stiffness of the active directions enters through the physical pseudodeterminant $\det_+K_*$. Exact null directions are mathematically excluded because the quadratic approximation does not confine the good set along them, while unresolved negative-curvature directions signify saddle points rather than local minima, rendering this local-basin construction valid exclusively when no resolved negative mode is present.

Within the locally nondegenerate regime of Theorem~\ref{thm:good_volume}, the Hessian-derived volume provides the leading approximation to the numerator of the useful-expressibility metric. Accordingly,
\begin{equation}
\label{eq:local_useful_expressibility}
f_\epsilon(\vec x_*) \approx \frac{V_{{\rm good},+}^{\rm FS}(\vec x_*;\epsilon)}
{\operatorname{Vol}_{\rm FS}(\mathcal M)}.
\end{equation}
This approximation applies when the stationary state is locally nondegenerate on its physical tangent space, so that every physically active tangent direction is confined to quadratic order and $r_+=r_{\rm loc}$. 

We first compare the normalized solution concentrations of constrained and unconstrained ansatze. Let $f_{\epsilon,\rm con}$ and $f_{\epsilon,\rm un}$ denote the useful expressibilities evaluated at their respective stationary states for the same target Hamiltonian and tolerance. Let $r_+^{\rm con}$ and $r_+^{\rm un}$ be the corresponding positive-curvature ranks, and let $K_{\rm con}$ and $K_{\rm un}$ be the associated Fubini--Study-normalized curvature operators. Defining
\begin{equation}
C_r
\equiv
\frac{\pi^{r/2}}{\Gamma(r/2+1)},
\end{equation}
and suppressing the asymptotically vanishing corrections in Theorem~\ref{thm:good_volume}, we obtain
\begin{equation}
\label{eq:volume_ratio_punchline_revised}
\resizebox{\columnwidth}{!}{$\displaystyle
\frac{f_{\epsilon,\rm con}}{f_{\epsilon,\rm un}}
\approx
\frac{\operatorname{Vol}_{\rm FS}(\mathcal M_{\rm un})}
{\operatorname{Vol}_{\rm FS}(\mathcal M_{\rm con})}
\frac{C_{r_+^{\rm con}}}{C_{r_+^{\rm un}}}
(2\epsilon)^{(r_+^{\rm con}-r_+^{\rm un})/2}
\left(
\frac{\det\nolimits_+K_{\rm un}}
{\det\nolimits_+K_{\rm con}}
\right)^{1/2}.
$}
\end{equation}

To separate this dimensional contribution from the local basin scale, we define the characteristic $\epsilon$-basin scale
\begin{equation}
\label{eq:characteristic_basin_scale}
R_\epsilon(\vec x_*)
\equiv
\left[
V_{{\rm good},+}^{\rm FS}(\vec x_*;\epsilon)
\right]^{1/r_+}.
\end{equation}
Invoking Theorem~\ref{thm:good_volume}, it is easy to see that this quantity satisfies
\begin{equation}
\label{eq:characteristic_basin_scale_explicit}
R_\epsilon(\vec x_*) = C_{r_+}^{1/r_+} \frac{\sqrt{2\epsilon}} {\left(\det\nolimits_+K_*\right)^{1/(2r_+)}} \left[1+o(1)\right]
\end{equation}
as $\epsilon\to0^+$. Note that we can equivalently write the empirical quantity $\frac{1}{r_+}\log V_{{\rm good},+}^{\rm FS}$ to be the logarithm of $R_\epsilon$. This quantity $R_\epsilon$ converts the $r_+$-dimensional local good volume into a volume-equivalent linear scale. For the locally ellipsoidal basin described by Theorem~\ref{thm:good_volume}, it is proportional to the geometric mean of the principal semi-axis lengths, with proportionality factor $C_{r_+}^{1/r_+}$, and therefore defines a characteristic basin scale.

At a common tolerance, the constrained-to-unconstrained comparison becomes
\begin{equation}
\label{eq:characteristic_basin_scale_ratio}
\frac{R_{\epsilon,\rm con}}
{R_{\epsilon,\rm un}}
\approx
\frac{C_{r_+^{\rm con}}^{1/r_+^{\rm con}}}
{C_{r_+^{\rm un}}^{1/r_+^{\rm un}}}
\frac{
\left(\det\nolimits_+K_{\rm un}\right)^{1/(2r_+^{\rm un})}
}{
\left(\det\nolimits_+K_{\rm con}\right)^{1/(2r_+^{\rm con})}
}.
\end{equation}
Consequently, Eq.~\eqref{eq:characteristic_basin_scale_ratio} compares the geometric-mean curvature scale of the retained confining directions without the power-law enhancement generated solely by the difference between $r_+^{\rm con}$ and $r_+^{\rm un}$.

The two diagnostics therefore answer distinct but complementary questions. The normalized useful expressibility $f_\epsilon$ measures how efficiently the full reachable physical manifold allocates its volume to the target-accurate neighborhood. Its constrained-to-unconstrained ratio includes the combined effects of global manifold contraction, positive-curvature rank, and local curvature. By contrast, $R_\epsilon$ measures the characteristic scale of the local accurate basin after accounting for the number of retained confining directions. A simultaneous increase in $f_\epsilon$ and $R_\epsilon$ demonstrates that symmetry projection does more than reduce the accessible dimension. It both concentrates the reachable manifold around target-relevant states and broadens the acceptable region within the physical directions that remain.

\section{Numerical implementation and reproducibility}
\label{sec:numerical_implementation}

All calculations were implemented in JAX using a custom codebase built on top of NetKet~\cite{jax2018github,netket2:2019,netket3:2022}. We considered periodic spin chains with all diagonal one- and two-body Pauli-$Z$ generators, corresponding to $k_{\max}=2$. The unconstrained ansatz was optimized over the amplitude coefficients $\{c_{\vec\alpha}\}$ and, when enabled, the phase coefficients $\{d_{\vec\alpha}\}$. Symmetry-constrained ansatze were optimized directly in reduced coordinates spanning the null space of the corresponding coefficient constraints. We considered the unconstrained, bit-flip, translation, reflection, space-group, and their combined symmetry families. The TFIM calculations optimized only the real amplitude sector, whereas the XXZ calculations included both amplitude and phase coefficients and employed the Marshall sign transformation~\cite{marshall1955antiferromagnetism}.

The $N=20$ TFIM and XXZ phase sweeps sampled $h/J\in[0,3]$ and $\Delta\in[-2,2]$, respectively, using denser grids near their critical regimes. Each run used $300$ stochastic-gradient iterations preconditioned by stochastic reconfiguration~\cite{sorella2005wave,carleo2017solving}, with learning rate $10^{-2}$ and diagonal shift $10^{-1}$. The TFIM used a local Metropolis sampler, while the XXZ used a nearest-neighbor exchange sampler~\cite{metropolis1953equation}. Both used $1008$ samples per iteration, $63$ persistent chains, a sweep size of $20$, and $20$ discarded samples per chain. Exact diagonalization, wavefunction fidelity, and real-space correlation diagnostics were evaluated for these $N=20$ benchmarks. Optimization used a fixed iteration budget without early stopping.

\begin{figure*}[t]
    \centering
    \includegraphics[
        width=\textwidth,
        trim={0 0 0 0},
        clip
    ]{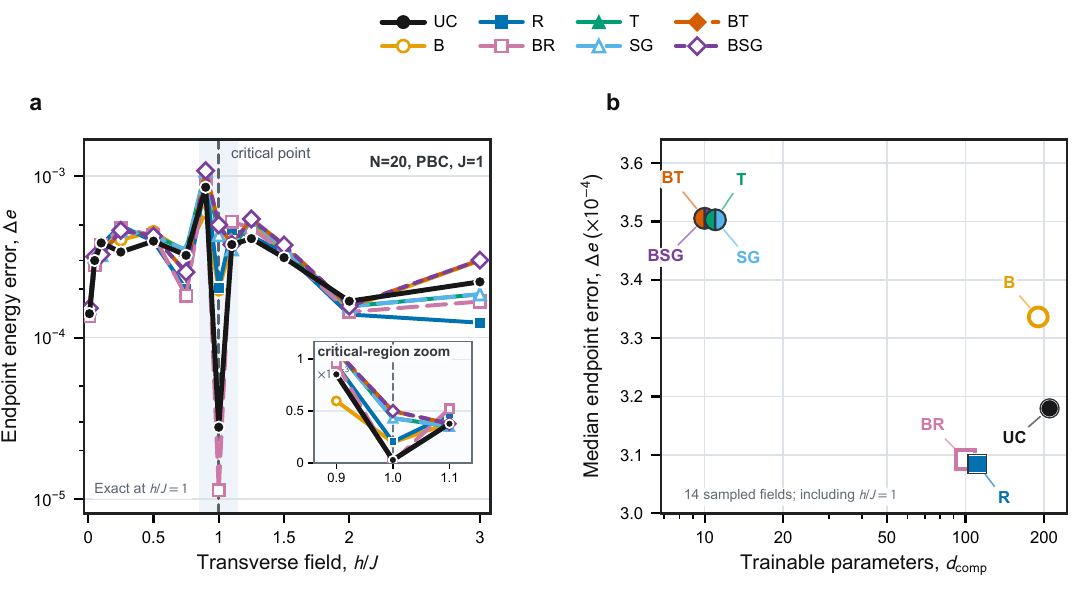}
    \caption{
    \textbf{Symmetry compression preserves TFIM variational accuracy.} 
    \textbf{a}, Endpoint energy-density error $\Delta e=|E_{\mathrm{VMC}}-E_{\mathrm{ref}}|/N$ for the 1D TFIM ($N=20$, $J=1$, $k_{\max}=2$). Each marker is the endpoint error of an independently optimized VMC run at a fixed transverse field $h$ and symmetry constraint. Lines are guides to the eye across the sweep, not optimization trajectories. The dashed vertical line marks the exact critical point $h_c/J=1$ of the one-dimensional TFIM~\cite{pfeuty1970tfim}.
    \textbf{b}, Sweep-level accuracy-compression summary. The horizontal axis indicates the number of trainable parameters after symmetry compilation ($d_{\mathrm{comp}}$), and the vertical axis is the median endpoint error over the field sweep. Translation-based constraints reduce the learner from 210 to 10--11 parameters while maintaining the unconstrained $10^{-4}$ error scale. Labels: UC (unconstrained); B (bitflip); R (reflection); BR (bitflip + reflection); T (translations); SG (space group); BT (bitflip + translations); BSG (bitflip + space group).
    }
    \label{fig:tfim_accuracy_compression}
\end{figure*}

Critical-TFIM geometry calculations were performed at the exact critical point $J=h=1$~\cite{pfeuty1970tfim} for $N\in\{8,12,16,20,32,64\}$ using eight independent seeds for each system size and symmetry family. These runs used the same optimizer settings for $300$ iterations. The local Metropolis sampler used $4096$ samples per iteration, $256$ persistent chains, a sweep size of $64$, $64$ discarded samples per chain, and a sampling chunk size of $32$.

Sampled Fubini--Study diagnostics used $1024$ state samples and eight Hutchinson probe vectors~\cite{hutchinson1990stochastic}. The local good-volume proxy was estimated from $128$ perturbations sampled uniformly from a parameter-space ball of radius $0.05$ around the optimized endpoint. The same perturbations were reused for $\epsilon_{\mathrm{dense}}\in\left\{10^{-4},3\times10^{-4},10^{-3},3\times10^{-3},10^{-2}\right\}$.

The global Fubini--Study normalization was estimated once for each system-size and symmetry-family pair using reference seed $0$, with $16$ outer samples drawn from a unit-radius parameter-space ball centered at the origin. Each log-pseudodeterminant estimate used matrix-free stochastic Lanczos quadrature~\cite{ubaru2017fast} with $4096$ state samples, eight probe vectors, $24$ Lanczos iterations, and full reorthogonalization. The sampled tangent-space rank cutoff was $\tau_{\mathrm{rank}}=10^{-10}$, while the absolute and relative Fubini--Study cutoffs were $\tau_{\mathrm{FS}}^{\mathrm{abs}}=10^{-12}$ and $\tau_{\mathrm{FS}}^{\mathrm{rel}}=10^{-8}$. These sampled calculations did not evaluate the objective Hessian and therefore did not require a Hessian-eigenvalue cutoff.

\section{Results}
\label{sec:results}

In this section, we present the empirical validation of the symmetry-compiled generalized NQS framework. We first examine the numerical tradeoffs between parameter compression and ground-state accuracy across standard one-dimensional spin chains. Following this, we evaluate the algorithmic scaling of optimization runtimes for larger systems. Finally, we map these performance gains back to their physical origins by analyzing the geometry and solution density of the resulting optimization landscapes.

\subsection{Variational accuracy and parameter compression}

We first evaluate whether symmetry compilation successfully compresses the parameter space without degrading the variational expressivity of the NQS. We benchmark against two standard one-dimensional systems with periodic boundary conditions, the transverse-field Ising model and the spin-$1/2$ XXZ chain~\cite{pfeuty1970tfim,yang1966xxzI,yang1966xxzIII}. Both models are evaluated at $N=20$ spins, an energy scale of $J=1$, and a locality cutoff $k_{\max}=2$. We track the energy-density error $\Delta e = |E_{\mathrm{VMC}}-E_{\mathrm{ref}}|/N$ to compare the optimized VMC energy against the exact ground state.

Figures~\ref{fig:tfim_accuracy_compression} and \ref{fig:xxz_accuracy_compression} summarize the accuracy-compression tradeoffs across the TFIM transverse-field ($h$) and XXZ anisotropy ($\Delta$) sweeps. We test eight symmetry constraints, ranging from the unconstrained baseline (UC) to the full one-dimensional space group augmented with bitflip symmetry (BSG).

The central finding is that hard-wiring symmetries yields substantial parameter reduction while maintaining variational accuracy within the resolution of the reported benchmarks. In both models, the unconstrained learner requires 210 trainable parameters. Imposing reflection or bitflip symmetries individually yields mild to moderate reductions. Enforcing translation-based constraints (T, SG, BT, BSG) reduces the learner to 10 to 11 parameters, corresponding to a roughly twenty-fold reduction.

Despite this severe dimensional truncation, the symmetry-compiled models do not systematically fail relative to the unconstrained baseline. Across both sweeps, the median endpoint errors remain tightly bounded in the $10^{-4}$ to $10^{-3}$ range. In our sweeps, the largest absolute errors occur near the TFIM critical point $h_c=1$ and near the strongly anisotropic edges of the investigated XXZ range, indicating that these parameter regimes are the most demanding under the fixed sampling and optimization budgets used here. Even in these demanding regimes, the heavily compressed learners remain within the same endpoint-error range as the unconstrained baseline.

These results confirm the physical premise of our construction: when the target Hamiltonian possesses a symmetry, the corresponding features in the ansatz are inherently redundant. Symmetry compilation structurally excises these redundancies, drastically reducing memory overhead without sacrificing the degrees of freedom required to capture the ground state.

\begin{figure*}[t]
    \centering
    \includegraphics[
        width=\textwidth,
        trim={0 0 0 0},
        clip
    ]{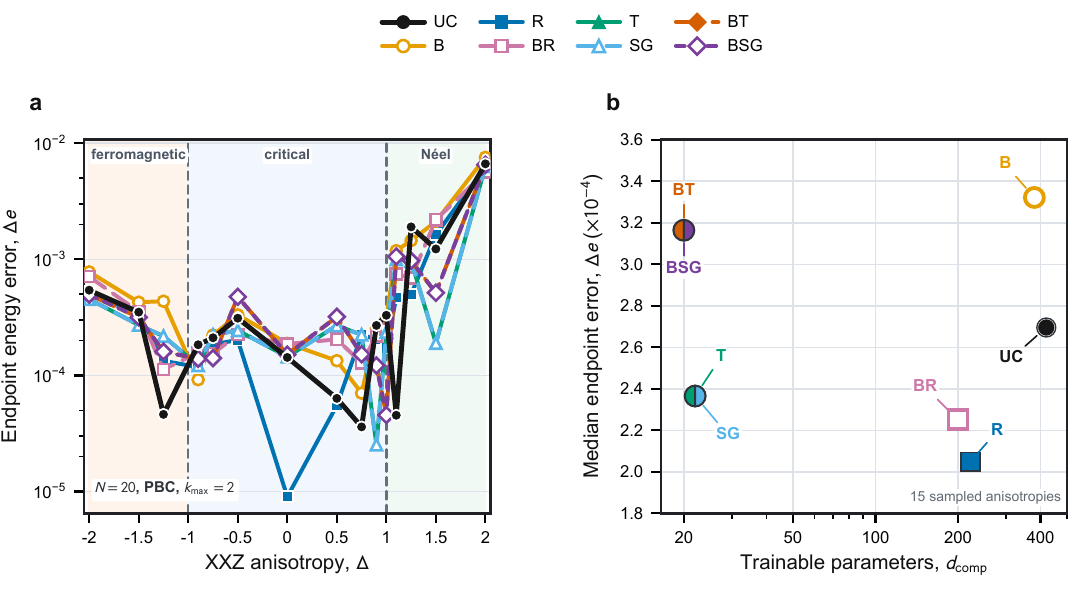}
    \caption{
    \textbf{Symmetry compression preserves XXZ variational accuracy across the anisotropy sweep.}
    \textbf{a}, Endpoint energy-density error $\Delta e=|E_{\mathrm{VMC}}-E_{\mathrm{ref}}|/N$ for the 1D XXZ benchmark ($N=20$, $J=1$, $k_{\max}=2$) as a function of anisotropy $\Delta$. Dashed vertical lines mark the standard reference anisotropies $\Delta=\pm 1$.
    \textbf{b}, Sweep-level accuracy-compression summary, sharing the axes and symmetry labels of Fig.~\ref{fig:tfim_accuracy_compression}. As in the TFIM benchmark, aggressive translation-based compression tightly matches the unconstrained median error. Labels: UC (unconstrained); B (bitflip); R (reflection); BR (bitflip + reflection); T (translations); SG (space group); BT (bitflip + translations); BSG (bitflip + space group).
    }
\label{fig:xxz_accuracy_compression}
\end{figure*}

\begin{figure*}[t]
    \centering
    \includegraphics[
        width=\textwidth,
        trim={0 0 0 0},
        clip
    ]{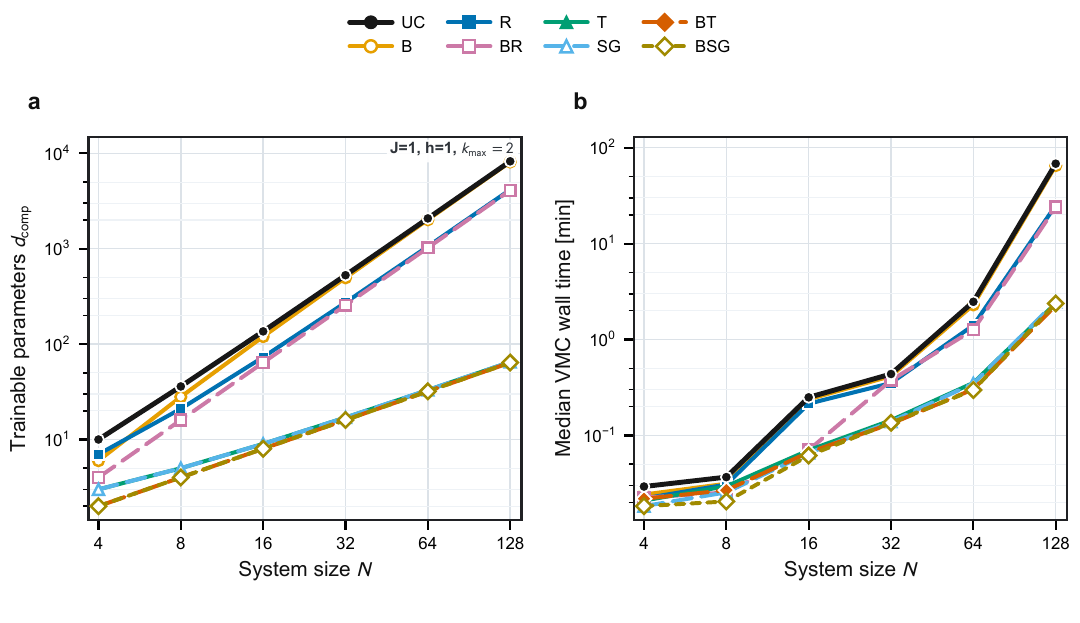}
    \caption{
    \textbf{Scaling of the symmetry-compiled learner.} \textbf{a}, Number of trainable parameters as a function of system size. \textbf{b}, VMC training wall time measured by the internal wall-clock timer around the optimization loop. Spatial symmetry constraints explicitly reduce the two-body learner from quadratic parameter growth to orbit-count growth, generating corresponding speedups in optimization runtime. Labels: UC (unconstrained); B (bitflip); R (reflection); BR (bitflip + reflection); T (translations); SG (space group); BT (bitflip + translations); BSG (bitflip + space group).
    }
    \label{fig:tfim-param-time-scaling}
\end{figure*}

\begin{table*}[t]
\centering
\caption{
\textbf{Diagnostic compression and runtime summaries for the largest TFIM sweeps.}
Compression and speedup are computed relative to the unconstrained learner at the equivalent system size. Here $T$ denotes $T_{\mathrm{VMC}}$, the wall time measured around the VMC optimization loop.
}
\label{tab:tfim-diagnostic-n64-n128}
\begingroup
\scriptsize
\setlength{\tabcolsep}{2.8pt}
\renewcommand{\arraystretch}{0.92}
\begin{minipage}[t]{0.485\textwidth}
\centering
\textbf{$N=64$}

\vspace{0.4em}

\begin{tabular}{lrrrrr}
\hline\hline
Con. & $d_{\mathrm{comp}}$ & Comp. & $T$ [min] & Speedup & $e_{\mathrm{best}}/N$ \\
\hline
UC  & 2080 & $1.00\times$  & $2.48$ & $1.00\times$ & $-1.272366$ \\
B   & 2016 & $1.03\times$  & $2.30$ & $1.08\times$ & $-1.272334$ \\
R   & 1057 & $1.97\times$  & $1.41$ & $1.76\times$ & $-1.272456$ \\
BR  & 1024 & $2.03\times$  & $1.26$ & $1.96\times$ & $-1.272157$ \\
T   & 33   & $63.03\times$ & $0.35$ & $7.00\times$ & $-1.272424$ \\
SG  & 33   & $63.03\times$ & $0.35$ & $7.11\times$ & $-1.272424$ \\
BT  & 32   & $65.00\times$ & $0.30$ & $8.17\times$ & $-1.272293$ \\
BSG & 32   & $65.00\times$ & $0.30$ & $8.32\times$ & $-1.272293$ \\
\hline\hline
\end{tabular}
\end{minipage}
\hfill
\begin{minipage}[t]{0.485\textwidth}
\centering
\textbf{$N=128$}

\vspace{0.4em}

\begin{tabular}{lrrrrr}
\hline\hline
Con. & $d_{\mathrm{comp}}$ & Comp. & $T$ [min] & Speedup & $e_{\mathrm{best}}/N$ \\
\hline
UC  & 8256 & $1.00\times$   & $68.21$ & $1.00\times$  & $-1.271802$ \\
B   & 8128 & $1.02\times$   & $65.19$ & $1.05\times$  & $-1.272023$ \\
R   & 4161 & $1.98\times$   & $25.05$ & $2.72\times$  & $-1.272167$ \\
BR  & 4096 & $2.02\times$   & $24.07$ & $2.83\times$  & $-1.272076$ \\
T   & 65   & $127.02\times$ & $2.39$  & $28.51\times$ & $-1.272099$ \\
SG  & 65   & $127.02\times$ & $2.47$  & $27.67\times$ & $-1.272099$ \\
BT  & 64   & $129.00\times$ & $2.27$  & $30.04\times$ & $-1.272026$ \\
BSG & 64   & $129.00\times$ & $2.37$  & $28.73\times$ & $-1.272026$ \\
\hline\hline
\end{tabular}
\end{minipage}

\vspace{0.35em}

\raggedright
\footnotesize
Constraint labels are UC (unconstrained), B (bitflip), R (reflection), BR (bitflip plus reflection), T (translations), SG (space group), BT (bitflip plus translations), and BSG (bitflip plus space group).
\endgroup
\end{table*}

\begin{figure*}[t]
    \centering
    \includegraphics[width=0.8\textwidth]{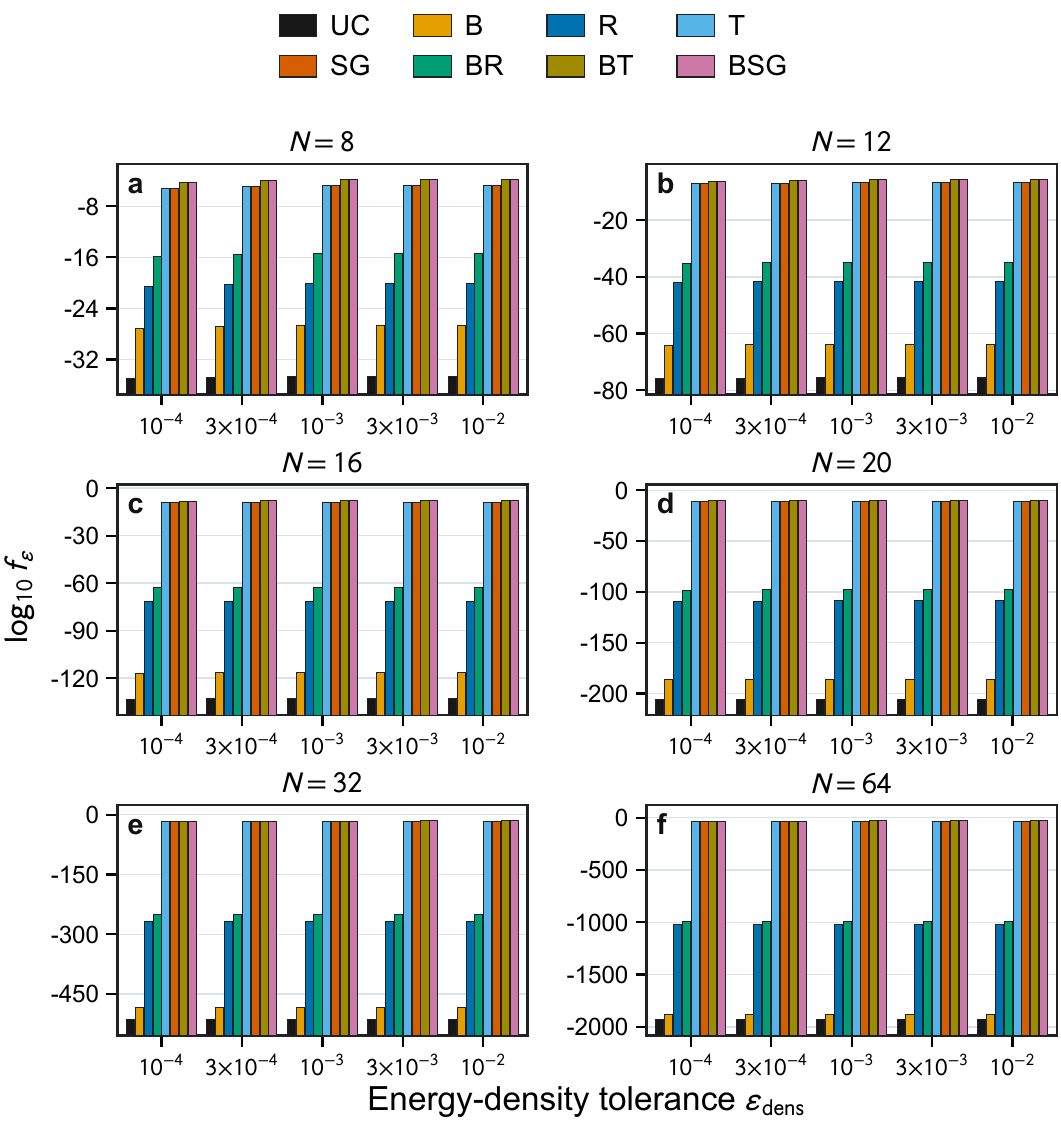}
    \caption{\textbf{Useful expressibility across system size and target tolerance.} The panels show the median useful expressibility $\log_{10}f_{\epsilon}$ for each system size $N$ as a function of the energy-density tolerance $\epsilon_{\mathrm{dens}}$. Here, $f_{\epsilon}=V_{\mathrm{good},+}^{\mathrm{FS}}(\epsilon)/\operatorname{Vol}_{\mathrm{FS}}(\mathcal{M})$ measures the fraction of the globally reachable variational manifold occupied by the local target-accurate region. Bars denote the unconstrained ansatz and the seven symmetry-constrained ansätze, with UC, B, R, T, SG, BR, BT, and BSG corresponding to unconstrained, bit-flip, reflection, translation, space-group, bit-flip--reflection, bit-flip--translation, and bit-flip--space-group constraints, respectively. Values are aggregated across independent seeds using the median. Within each tolerance group, bars are arranged in increasing order of height while their colors remain fixed by symmetry class.}
\label{fig:useful_expressibility_epsilon}
\end{figure*}

\subsection{Symmetry constraints improve training robustness}
\label{subsec:symmetry-training-robustness}

In the preceding subsection, we showed that physically motivated symmetry constraints maintain ground-state accuracy while substantially reducing the parameter count and training time for both the transverse-field Ising and XXZ models. We now test whether these computational gains are accompanied by a measurable regularization of the underlying optimization geometry. We use the target-aware useful-expressibility measures introduced in Definition~\ref{def:trainability_metric} and evaluate the critical transverse-field Ising model at $h=J=1$ for system sizes $N\in\{8,12,16,20,32,64\}$ and energy-density tolerances $\epsilon_{\mathrm{dens}}\in\{10^{-4},3\times10^{-4},10^{-3},3\times10^{-3},10^{-2}\}$. We compare the unconstrained neural quantum state with ansätze constrained by bit-flip, reflection, translation, and space-group symmetries, together with their combined actions. The central geometric result is summarized in Figs.~\ref{fig:useful_expressibility_epsilon} and~\ref{fig:characteristic_basin_width_epsilon}. Figure~\ref{fig:useful_expressibility_epsilon} measures the concentration of target-accurate states within the reachable manifold, whereas Fig.~\ref{fig:characteristic_basin_width_epsilon} measures the characteristic scale of the accurate basin along its retained confining directions. Together, these figures show that \textit{symmetry constraints improve both the global concentration and the local geometry of target-accurate states.}

Figure~\ref{fig:useful_expressibility_epsilon} reports $\log_{10}f_{\epsilon}$, where, by Definition~\ref{def:trainability_metric}, $f_{\epsilon}$ is the fraction of the reachable Fubini--Study volume contained in the connected $\epsilon$-accurate neighborhood of the target stationary state. The numerator and denominator entering each value of $f_{\epsilon}$ are evaluated on the same physical state manifold and therefore have the same intrinsic dimension. The values shown in Fig.~\ref{fig:useful_expressibility_epsilon} consequently represent dimensionless concentrations of target-accurate states within the corresponding reachable manifolds.

Across every system size and tolerance shown in Fig.~\ref{fig:useful_expressibility_epsilon}, each symmetry-constrained ansatz has a larger value of $f_{\epsilon}$ than the unconstrained ansatz. The ordering is also remarkably stable. The unconstrained manifold consistently has the smallest useful-volume fraction, while the bit-flip-constrained manifold remains closest to it. The reflection and bit-flip--reflection manifolds form an intermediate band. The translation, space-group, bit-flip--translation, and bit-flip--space-group manifolds form a clearly separated upper band with the largest useful-volume fractions. The close agreement among the translation-containing sectors indicates that translational invariance supplies the dominant contribution to the increase in target concentration, while the addition of reflection or bit-flip symmetry produces smaller corrections.

The constrained-to-unconstrained separation visible in Fig.~\ref{fig:useful_expressibility_epsilon} increases strongly with $N$, while the ordering changes only weakly as $\epsilon_{\mathrm{dens}}$ is varied from $10^{-4}$ to $10^{-2}$. This behavior shows that the observed separation is neither a finite-size peculiarity nor an artifact of selecting a single accuracy threshold. Instead, the unconstrained ansatz allocates an increasingly small fraction of its reachable physical capacity to the target neighborhood as the system grows. Symmetry-aware coefficient tying reverses this dilution by restricting the reachable manifold to a substantially more target-aligned set of physical states. This geometric concentration is precisely what $f_{\epsilon}$ was designed to resolve and cannot be inferred from the raw parameter count alone.

However, as iterated before, an increase in $f_{\epsilon}$ alone does not establish that the target-accurate region becomes well conditioned along its retained physical directions. Figure~\ref{fig:characteristic_basin_width_epsilon} therefore separates the local basin scale from the normalized global concentration by reporting $\log_{10}R_{\epsilon}$, where $R_{\epsilon}$ is the characteristic $\epsilon$-basin scale defined in Eq.~\eqref{eq:characteristic_basin_scale}.

\begin{figure*}[t]
\centering
\includegraphics[width=0.8\textwidth]{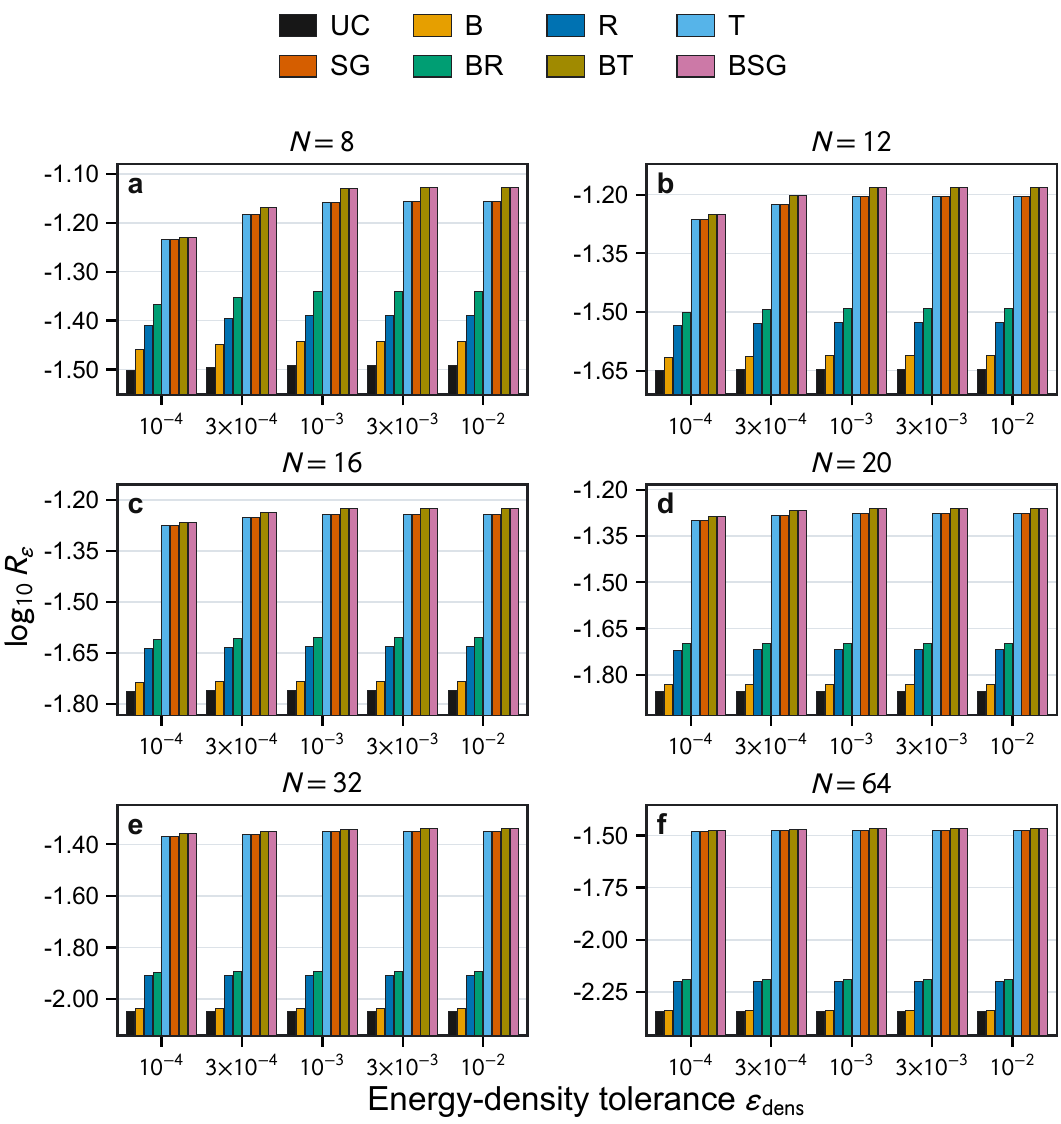}
\caption{\textbf{Characteristic width of the target-accurate parameter basin.} The panels show the median logarithmic characteristic basin width $\log_{10}R_{\epsilon}$ for each system size $N$ as a function of the energy-density tolerance $\epsilon_{\mathrm{dens}}$, where $R_{\epsilon}=\left[V_{\mathrm{good},+}^{\mathrm{FS}}(\epsilon)\right]^{1/r_{+}}$ and $r_{+}$ is the dimension of the positive-curvature physical subspace. This rank-normalized quantity converts the local Fubini--Study volume into an effective linear scale and therefore permits comparison between variational manifolds of different physical dimension. Bars denote the unconstrained ansatz and the seven symmetry-constrained ansätze, with UC, B, R, T, SG, BR, BT, and BSG corresponding to unconstrained, bit-flip, reflection, translation, space-group, bit-flip--reflection, bit-flip--translation, and bit-flip--space-group constraints, respectively. Values are aggregated across independent seeds using the median. Within each tolerance group, bars are arranged in increasing order of height while their colors remain fixed by symmetry class.}
\label{fig:characteristic_basin_width_epsilon}
\end{figure*}

The results in Fig.~\ref{fig:characteristic_basin_width_epsilon} show that the increase in useful expressibility is not produced solely by contraction of the total manifold or by a reduction in the number of confining directions. The unconstrained and bit-flip-only manifolds exhibit the smallest values of $R_{\epsilon}$ and narrow systematically with increasing $N$. Reflection and bit-flip--reflection occupy an intermediate regime. By contrast, the translation-containing constraints retain substantially larger characteristic basin scales throughout the investigated range. In particular, $\log_{10}R_{\epsilon}$ for the translation and space-group sectors remains approximately between $-1.2$ and $-1.5$, whereas the unconstrained value decreases from approximately $-1.5$ at $N=8$ to below $-2.3$ at $N=64$. Thus, after expressing the local good volume on the rank-normalized linear scale defined by $R_{\epsilon}$, the translation-constrained manifolds continue to exhibit broader target-accurate basins than the unconstrained manifold.

The agreement between Figs.~\ref{fig:useful_expressibility_epsilon} and~\ref{fig:characteristic_basin_width_epsilon} is central to the geometric interpretation. Figure~\ref{fig:useful_expressibility_epsilon} shows that symmetry constraints increase the fraction of the reachable manifold occupied by target-accurate states, while Fig.~\ref{fig:characteristic_basin_width_epsilon} shows that the corresponding local basins remain broader after the good volume is converted to the rank-normalized scale $R_{\epsilon}$. The enhancement of $f_{\epsilon}$ therefore cannot be explained by the reduction in positive-curvature rank alone. Instead, symmetry constraints improve the geometry in two distinct ways. They concentrate the reachable manifold around the target neighborhood and reduce the effective stiffness of the confining directions that remain. The common ordering across both figures indicates that translation-containing constraints produce the strongest improvement in both global target concentration and local basin geometry.

The constrained-to-unconstrained separation visible in Fig.~\ref{fig:useful_expressibility_epsilon} contains the combined effects identified in Eq.~\eqref{eq:volume_ratio_punchline_revised}. Global contraction of the constrained manifold, changes in positive-curvature rank, and reconditioning of the retained curvature spectrum all contribute to the observed differences in $\log_{10}f_{\epsilon}$. The useful-volume fraction should therefore not be interpreted as an isolated measure of basin broadening. Local broadening is resolved separately in Fig.~\ref{fig:characteristic_basin_width_epsilon}, where the characteristic basin scale $R_{\epsilon}$ remains substantially larger for translation-containing constraints. The componentwise diagnostics supporting this interpretation are provided in Appendix~\ref{app:geometry-diagnostics}, where we separately resolve the positive-curvature and physical tangent-space contributions underlying the main geometric measures. The persistence of this hierarchy across optimization checkpoints and scalable large-system estimators is established independently in Appendices~\ref{app:geometry-training-robustness} and~\ref{app:sampled-geometry-proxies}.

Taken together, Figs.~\ref{fig:useful_expressibility_epsilon} and~\ref{fig:characteristic_basin_width_epsilon} provide direct empirical support for the useful-expressibility framework. Within the evaluated TFIM family, the investigated tolerance range, and the local positive-curvature construction of Theorem~\ref{thm:good_volume}, symmetry compilation does not merely reduce the number of trainable coordinates. It concentrates the reachable physical manifold around the target state, preserves a larger characteristic basin scale along the retained confining directions, and produces a normalized useful-volume advantage that grows rapidly with system size. The agreement of these geometric diagnostics with the accuracy-preserving runtime improvements established above identifies this reorganization as the geometric mechanism underlying the observed training robustness.

\section{Conclusion}
\label{sec:conclusion}

The main takeaway of this work is that physical symmetries are most useful when they are built into the variational model before optimization begins. For the neural quantum states studied here, imposing constraints using the symmetry of the physical model directly into coefficient space produces substantially smaller ansatze that retain the accuracy of the unconstrained model while requiring far fewer trainable degrees of freedom. More importantly, this compression does not simply remove parameters, rather it reorganizes the reachable physical state space so that the remaining variational capacity is better aligned with the target problem. The geometric framework introduced in this work makes this reorganization quantitative.

To quantify this reorganization, we introduced a target-aware geometric notion of expressibility based on the Fubini--Study volume of the reachable state manifold. The useful fraction $f_\epsilon$ quantifies how much of that manifold lies within a prescribed accuracy of the target, while the characteristic basin scale $R_\epsilon$ resolves the local extent of the same region after accounting for its positive-curvature rank. Taken together, these measures show whether symmetry has genuinely improved the target geometry, rather than merely shrinking the variational space, and they allow constrained and unconstrained ansatze to be compared on the same physical footing.

The numerical results support this picture across the TFIM and XXZ benchmarks. Symmetry-constrained ansatze preserve the relevant ground-state observables and endpoint energy accuracy while using substantially fewer parameters and reducing the computational cost of training. The same hierarchy appears in the geometric diagnostics. Translation- and space-group-based constraints produce the largest increases in both $f_\epsilon$ and $R_\epsilon$, while reflection-based constraints provide a more moderate improvement and bit-flip symmetry alone yields the smallest gain. The separation from the unconstrained ansatz also grows with system size, showing that the geometric advantage of symmetry compilation becomes more pronounced as the variational problem scales.

These gains follow from an explicit symmetry-compilation procedure rather than from ad hoc pruning. The action of the physical symmetry group on local Pauli supports identifies coefficient orbits that can be represented by a common variational coordinate, while Burnside counting determines the number of independent parameters that survive before training begins. This construction, making the compression transparent and reproducible.

More broadly, these results suggest that variational expressibility should be judged by where the reachable states lie, not simply by how many states the ansatz can represent. Symmetry compilation provides a practical way to enforce this principle by concentrating variational freedom within physically relevant sectors before optimization begins. The combination of orbit-based parameter reduction with the geometric diagnostics $f_\epsilon$ and $R_\epsilon$ therefore offers a systematic framework for designing and comparing structured ansatze. Although developed here for Boltzmann-family neural quantum states, the same viewpoint is applicable to other variational representations in which symmetries can be compiled into the model coordinates.

The framework developed here opens a broader route toward symmetry-aware variational learning beyond the specific NQS constructions studied in this work. The same geometric perspective can be extended to higher-dimensional lattices, non-Abelian and internal symmetries, excited-state problems, and time-dependent variational dynamics. More generally, combining symmetry compilation with target-aware geometric diagnostics provides a principled strategy for designing ansätze whose expressibility is concentrated where it is physically useful. This shifts ansatz design away from maximizing raw variational freedom and toward organizing that freedom around the structure of the target problem.

\section{Data Availability}
The source code and run configurations used to generate the reported results are available from the authors upon reasonable request.

\section{Acknowledgements}
This work was supported by the Quantum Science Center, a National Quantum Information Science Research Center of the U.S. Department of Energy (DOE), operated at Oak Ridge National Laboratory (ORNL). EL and BNB were supported by the U.S. Department of Energy, Advanced Scientific Computing Research, under contract number DE-SC0025384. M.S would like to acknowledge the use of resources of the Oak Ridge Leadership Computing Facility at the Oak Ridge
National Laboratory, which is supported by the Office of Science of the U.S. Department of Energy under Contract No. DE-AC05-00OR22725. This manuscript has in part been authored by UT-Battelle, LLC under Contract No. DE-AC05-
00OR22725 with the U.S. Department of Energy. The United States Government retains
and the publisher, by accepting the article for publication, acknowledges that the U.S.
Government retains a non-exclusive, paid up, irrevocable, world-wide license to publish
or reproduce the published form of the manuscript, or allow others to do so, for U.S.
Government purposes. The Department of Energy will provide public access to these
results of federally sponsored research in accordance with the DOE Public Access Plan
\href{http://energy.gov/downloads/doe-publicaccess-plan}{(see link)}.


\bibliography{ref}

\appendix
\onecolumngrid

\section{Proofs of Theorems in Section~\ref{sec:trainability_nqs}}
\label{app:jacobian_proofs}


\begin{proof}[Proof of Lemma~\ref{lem:uncon_jacobian}]
The unconstrained amplitude is $\psi(\vec z;\vec c,\vec d) = \mathcal Z(\vec c)^{-1/2} e^{-\frac{\beta}{2}h(\vec c,\vec z)} e^{-ig(\vec d,\vec z)}$. From $h(\vec c,\vec z) = \sum_{\vec\alpha}(-1)^{\vec\alpha\cdot\vec z}c_{\vec\alpha}$ and $g(\vec d,\vec z) = \sum_{\vec\alpha}(-1)^{\vec\alpha\cdot\vec z}d_{\vec\alpha}$,
\begin{equation}
\frac{\partial h(\vec c,\vec z)}{\partial c_{\vec\beta}} = (-1)^{\vec\beta\cdot\vec z}, \qquad \frac{\partial g(\vec d,\vec z)}{\partial c_{\vec\beta}} = 0.
\end{equation}
Differentiating $\psi$ with respect to $c_{\vec\beta}$ by the product rule,
\begin{equation}
\label{eq:app_psi_deriv}
\frac{\partial\psi(\vec z;\vec c,\vec d)}{\partial c_{\vec\beta}} = \frac{\partial\mathcal Z(\vec c)^{-1/2}}{\partial c_{\vec\beta}}\,e^{-\frac{\beta}{2}h(\vec c,\vec z)}e^{-ig(\vec d,\vec z)} + \mathcal Z(\vec c)^{-1/2}\frac{\partial e^{-\frac{\beta}{2}h(\vec c,\vec z)}}{\partial c_{\vec\beta}}e^{-ig(\vec d,\vec z)}.
\end{equation}
The second term evaluates directly to $-\frac{\beta}{2}(-1)^{\vec\beta\cdot\vec z}\psi(\vec z;\vec c,\vec d)$. For the first term, differentiating the partition function $\mathcal Z(\vec c) = \sum_{\vec z'}e^{-\beta h(\vec c,\vec z')}$ gives
\begin{equation}
\frac{1}{\mathcal Z(\vec c)}\frac{\partial\mathcal Z(\vec c)}{\partial c_{\vec\beta}} = -\beta\sum_{\vec z'}\frac{e^{-\beta h(\vec c,\vec z')}}{\mathcal Z(\vec c)}(-1)^{\vec\beta\cdot\vec z'} = -\beta\big\langle(-1)^{\vec\beta\cdot\vec z}\big\rangle_\psi,
\end{equation}
where $\langle\cdot\rangle_\psi$ is the expectation over the Born distribution $|\psi(\vec z;\vec c,\vec d)|^2$. Since $\partial_{c_{\vec\beta}}\mathcal Z^{-1/2} = -\tfrac12\mathcal Z^{-3/2}\partial_{c_{\vec\beta}}\mathcal Z$, the first term of Eq.~\eqref{eq:app_psi_deriv} becomes $\frac{\beta}{2}\psi(\vec z;\vec c,\vec d)\big\langle(-1)^{\vec\beta\cdot\vec z}\big\rangle_\psi$. Combining both terms gives
\begin{equation}
\frac{\partial\psi(\vec z;\vec c,\vec d)}{\partial c_{\vec\beta}} = -\frac{\beta}{2}\psi(\vec z;\vec c,\vec d)\Big[(-1)^{\vec\beta\cdot\vec z} - \big\langle(-1)^{\vec\beta\cdot\vec z}\big\rangle_\psi\Big].
\end{equation}

For the phase derivative, $h(\vec c,\vec z)$ and $\mathcal Z(\vec c)$ are independent of $d_{\vec\beta}$, while $\partial g(\vec d,\vec z)/\partial d_{\vec\beta} = (-1)^{\vec\beta\cdot\vec z}$, so
\begin{equation}
\frac{\partial\psi(\vec z;\vec c,\vec d)}{\partial d_{\vec\beta}} = \mathcal Z(\vec c)^{-1/2}e^{-h(\vec c,\vec z)}\frac{\partial e^{-ig(\vec d,\vec z)}}{\partial d_{\vec\beta}} = -i(-1)^{\vec\beta\cdot\vec z}\psi(\vec z;\vec c,\vec d),
\end{equation}
which establishes both Jacobian identities.
\end{proof}


\begin{proof}[Proof of Lemma~\ref{lem:con_jacobian}]
For a generic real parameter $\theta \in \{\xi_a^R,\xi_a^I\}$, the same product-rule expansion as above gives the general identity
\begin{equation}
    \label{eq:app_general_constraint_id}
    \frac{\partial\psi(\vec z;\vec\xi)}{\partial\theta} = -\psi(\vec z;\vec\xi)\left[\frac{\beta}{2}\left(\frac{\partial h(\vec\xi,\vec z)}{\partial\theta} - \left\langle\frac{\partial h}{\partial\theta}\right\rangle_\psi\right) + i\,\frac{\partial g(\vec\xi,\vec z)}{\partial\theta}\right],
\end{equation}
where $\mathcal Z(\vec\xi)^{-1}\partial_\theta\mathcal Z(\vec\xi)=-\beta\langle\partial_\theta h\rangle_\psi$ follows exactly as in the unconstrained case, now applied to $\mathcal Z(\vec\xi) = \sum_{\vec z'}e^{-\beta h(\vec\xi,\vec z')}$.

Setting $\theta = \xi_a^R$, the constrained eigenvalues $h(\vec\xi,\vec z) = \sum_a(A_{\vec z,a}^R\xi_a^R - A_{\vec z,a}^I\xi_a^I)$ and $g(\vec\xi,\vec z) = \sum_a(A_{\vec z,a}^R\xi_a^I + A_{\vec z,a}^I\xi_a^R)$ give
\begin{equation}
    \frac{\partial h(\vec\xi,\vec z)}{\partial\xi_a^R} = A_{\vec z,a}^R, \qquad \frac{\partial g(\vec\xi,\vec z)}{\partial\xi_a^R} = A_{\vec z,a}^I,
\end{equation}
so that $\langle\partial_{\xi_a^R}h\rangle_\psi = \langle A_a^R\rangle_\psi$. Substituting into Eq.~\eqref{eq:app_general_constraint_id},
\begin{equation}
    \frac{\partial\psi(\vec z;\vec\xi)}{\partial\xi_a^R} = -\psi(\vec z;\vec\xi)\left[\frac{\beta}{2}\Big(A_{\vec z,a}^R - \langle A_a^R\rangle_\psi\Big) + iA_{\vec z,a}^I\right].
\end{equation}

Setting $\theta = \xi_a^I$ instead gives $\partial_{\xi_a^I}h = -A_{\vec z,a}^I$ and $\partial_{\xi_a^I}g = A_{\vec z,a}^R$, hence $\langle\partial_{\xi_a^I}h\rangle_\psi = -\langle A_a^I\rangle_\psi$, and Eq.~\eqref{eq:app_general_constraint_id} yields
\begin{equation}
    \frac{\partial\psi(\vec z;\vec\xi)}{\partial\xi_a^I} = -\psi(\vec z;\vec\xi)\left[-\frac{\beta}{2}\Big(A_{\vec z,a}^I - \langle A_a^I\rangle_\psi\Big) + iA_{\vec z,a}^R\right].
\end{equation}
This proves both constrained Jacobian identities.
\end{proof}


\begin{proof}[Proof of Lemma~\ref{lem:hessian}]
Differentiating $E(\vec\theta) = \langle\psi(\vec\theta)|\mathcal H|\psi(\vec\theta)\rangle$ once with respect to $\theta_\mu$,
\begin{equation}
    \frac{\partial E}{\partial\theta_\mu} = \langle\partial_\mu\psi|\mathcal H|\psi\rangle + \langle\psi|\mathcal H|\partial_\mu\psi\rangle = 2\,\mathrm{Re}\big[\langle\partial_\mu\psi|\mathcal H|\psi\rangle\big],
\end{equation}
using $\langle\psi|\mathcal H|\partial_\mu\psi\rangle = \big[\langle\partial_\mu\psi|\mathcal H|\psi\rangle\big]^*$, which follows from $\mathcal H = \mathcal H^\dagger$. Differentiating again with respect to $\theta_\nu$,
\begin{equation}
    \frac{\partial^2 E}{\partial\theta_\mu\partial\theta_\nu} = \langle\partial_\mu\partial_\nu\psi|\mathcal H|\psi\rangle + \langle\partial_\mu\psi|\mathcal H|\partial_\nu\psi\rangle + \langle\partial_\nu\psi|\mathcal H|\partial_\mu\psi\rangle + \langle\psi|\mathcal H|\partial_\mu\partial_\nu\psi\rangle,
\end{equation}
and pairing conjugate terms via Hermiticity of $\mathcal H$ gives
\begin{equation}
    \frac{\partial^2 E}{\partial\theta_\mu\partial\theta_\nu} = 2\,\mathrm{Re}\big[\langle\partial_\mu\psi|\mathcal H|\partial_\nu\psi\rangle + \langle\partial_\mu\partial_\nu\psi|\mathcal H|\psi\rangle\big].
\end{equation}
Writing $J$ for the Jacobian with columns $|\partial_\mu\psi\rangle$ and $R_{\mu\nu} = \langle\partial_\mu\partial_\nu\psi|\mathcal H|\psi\rangle$, this is compactly $\nabla^2_{\vec\theta}E = 2\,\mathrm{Re}[J^\dagger\mathcal H J + R]$, as claimed.
\end{proof}


\begin{proof}[Proof of Theorem~\ref{thm:good_volume}]
Let $S_*=V\Lambda V^\top$ be the spectral decomposition on its rank-$r_{\rm loc}$ support, with $V\in\mathbb R^{d\times r_{\rm loc}}$, and introduce metric-normalized tangent coordinates $\vec y=\Lambda^{1/2}V^\top d\vec\theta$; in these coordinates the physical curvature is represented by $\widetilde K_*=\Lambda^{-1/2}V^\top M_*V\Lambda^{-1/2}\in\mathbb R^{r_{\rm loc}\times r_{\rm loc}}$, whose nonzero eigenvalues coincide with those of the parameter-space matrix $K_*$ in Eq.~\eqref{eq:normalized_curvature}, and we henceforth identify $K_*$ with this induced operator on $T_{\vec x_*}\mathcal M$.

Let $\{\vec e_i\}_{i\in\mathcal I_+}$ be a Fubini--Study-orthonormal eigenbasis of the positive-curvature physical subspace $\mathcal T_+$, where $K_*\vec e_i=\mu_i\vec e_i$ with $\mu_i>\tau_{\rm H}$. Any tangent vector $\vec u\in\mathcal T_+$ can be written as $\vec u=\sum_{i\in\mathcal I_+}y_i\vec e_i$, where $\vec y\in\mathbb R^{r_+}$.

Using the Fubini--Study exponential map, write a nearby physical state as $\vec x=\exp_{\vec x_*}(\vec u)$. Since $\vec x_*$ is stationary, the energy expansion restricted to $\mathcal T_+$ is
\begin{equation}
    E\!\left(\exp_{\vec x_*}(\vec u)\right)-E(\vec x_*)=\frac{1}{2}\langle\vec u,K_*\vec u\rangle_{\rm FS}+O(\|\vec u\|_{\rm FS}^3)=\frac{1}{2}\sum_{i\in\mathcal I_+}\mu_i y_i^2+O(\|\vec y\|^3).
\end{equation}
Keeping terms to quadratic order, the condition defining the local $\epsilon$-good region becomes
\begin{equation}
    \frac{1}{2}\sum_{i\in\mathcal I_+}\mu_i y_i^2\le\epsilon,
\end{equation}
or equivalently
\begin{equation}
\sum_{i\in\mathcal I_+}\mu_i y_i^2\le2\epsilon.
\end{equation}
Because $K_*$ is positive definite on $\mathcal T_+$, this condition defines an $r_+$-dimensional ellipsoid with semi-axis lengths
\begin{equation}
    a_i=\sqrt{\frac{2\epsilon}{\mu_i}},\qquad i\in\mathcal I_+.
\end{equation}
In the Fubini--Study-orthonormal coordinates $\vec y$, the induced volume element satisfies
\begin{equation}
{\rm dVol}_{\rm FS}=\left[1+o(1)\right]\prod_{i\in\mathcal I_+}dy_i
\end{equation}
as $\epsilon\to0^+$. Therefore, the Fubini--Study volume of the local good region is the volume of the $r_+$-dimensional unit ball multiplied by the product of the ellipsoid semi-axis lengths:
\begin{align}
    V_{{\rm good},+}^{\rm FS}(\vec x_*;\epsilon)&\equiv\operatorname{Vol}_{\rm FS}\!\left[\mathcal G_\epsilon(\vec x_*)\cap\exp_{\vec x_*}(\mathcal T_+)\right]\nonumber\\
    &=\frac{\pi^{r_+/2}}{\Gamma(r_+/2+1)}\prod_{i\in\mathcal I_+}a_i\left[1+o(1)\right]\nonumber\\
    &=\frac{\pi^{r_+/2}}{\Gamma(r_+/2+1)}\frac{(2\epsilon)^{r_+/2}}{\sqrt{\prod_{i\in\mathcal I_+}\mu_i}}\left[1+o(1)\right].
\end{align}
Using $\det\nolimits_+K_*=\prod_{i\in\mathcal I_+}\mu_i$, we obtain
\begin{equation}
    V_{{\rm good},+}^{\rm FS}(\vec x_*;\epsilon)=\frac{\pi^{r_+/2}}{\Gamma(r_+/2+1)}\frac{(2\epsilon)^{r_+/2}}{\sqrt{\det\nolimits_+K_*}}\left[1+o(1)\right]
\end{equation}
as $\epsilon\to0^+$, which proves the result.
\end{proof}

\section{Fidelity and Real-Space Correlation Diagnostics}
\label{app:fidelity-correlation-probes}

While low variational energy indicates algorithmic convergence, it does not strictly guarantee correct wavefunction overlap or real-space ordering. To verify that the symmetry-projected manifolds capture the physical target states, we supplement the main energy diagnostics with the exact wavefunction infidelity $1-\mathcal F$ and real-space two-point correlation functions for both the TFIM and XXZ benchmarks. 

For a periodic chain of length $N$, the translation-averaged spin correlator is
\begin{equation}
C_{ZZ}(r)=\frac{1}{N}\sum_{i=1}^{N}\langle Z_i Z_{i+r}\rangle.
\end{equation}
For the ferromagnetic TFIM convention $H=-J\sum_i Z_iZ_{i+1}-h\sum_iX_i$ with $J>0$, the low-field longitudinal correlations do not alternate in sign. We therefore use the translation-averaged correlator $C_{ZZ}(r)$ directly. For the XXZ model, we use the corresponding correlation profile $C(r)$. To provide a compact scalar probe of the learned real-space structure across the phase sweeps, we define the distance-averaged correlation strength
\begin{equation}
\overline{C}=\frac{1}{R}\sum_{r=1}^{R}C(r),
\end{equation}
where $R$ is the maximum measured spatial separation and $C(r)=C_{ZZ}(r)$ for the TFIM.

\begin{figure*}[t]
    \centering
    \includegraphics[width=\textwidth, trim={0 0 0 0}, clip]{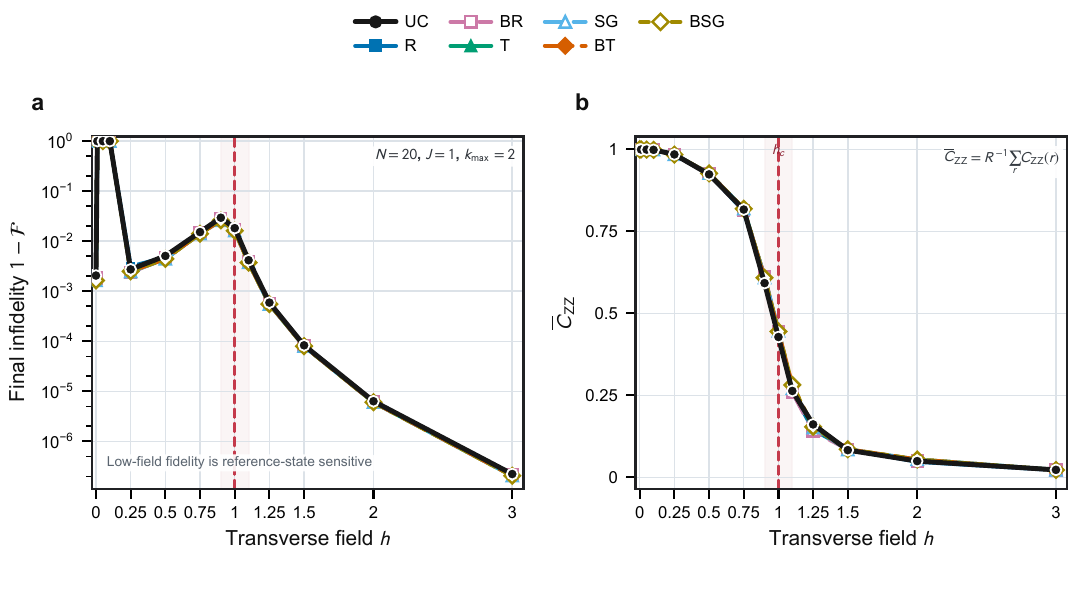}
    \caption{\textbf{Critical-field correlation probe for the TFIM at $h/J=1$.} \textbf{a}, The endpoint longitudinal correlation profile $C_{ZZ}(r)$ demonstrates that symmetry tying preserves the learned spatial correlations across all measured separations. \textbf{b}, The training-time growth of $\overline{C}_{ZZ}$ confirms that these correlations are built dynamically during optimization.}
    \label{fig:app-nowarm-fid-cstag-sweep}
\end{figure*}

\begin{figure*}[t]
    \centering
    \includegraphics[width=\textwidth, trim={0 0 0 0}, clip]{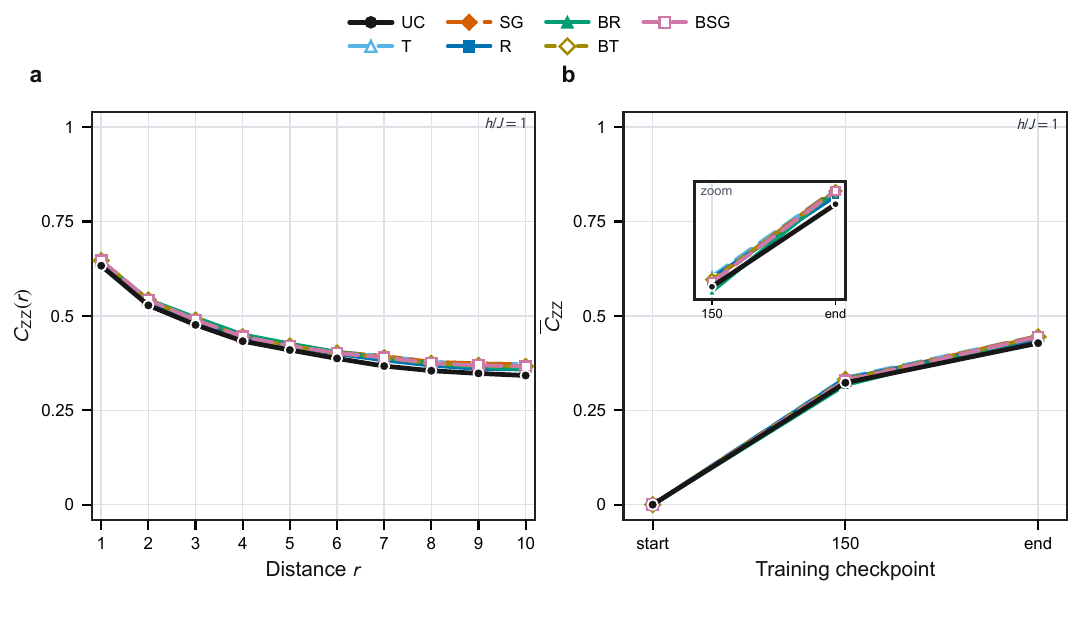}
    \caption{Critical-field correlation probe for the TFIM at $h/J=1$. (a) The endpoint staggered correlation profile $C_{\rm stag}(r)$ demonstrates that symmetry tying perfectly preserves the learned spatial order across all separations. (b) The training-time growth of $\overline C_{\rm stag}$ confirms these correlations are built dynamically during optimization.}
    \label{fig:app-nowarm-cstag-profile-growth}
\end{figure*}

For the TFIM (Figs.~\ref{fig:app-nowarm-fid-cstag-sweep} and~\ref{fig:app-nowarm-cstag-profile-growth}), the infidelity and spatial correlation diagnostics for the constrained networks are virtually indistinguishable from the unconstrained baseline. The localized elevation in $1-\mathcal F$ deep in the ordered phase is a known artifact of evaluating scalar fidelity against a single reference vector within a nearly degenerate subspace; it penalizes physically identical states that differ merely by relative phase within the doublet. The agreement of the longitudinal correlation observable $\overline{C}_{ZZ}$ supports this interpretation. At the critical field $h/J=1$, the spatial profiles $C_{\rm stag}(r)$ overlap perfectly, demonstrating that symmetry projection does not eliminate the variational directions required to capture extended spatial structure.

\begin{figure*}[t]
    \centering
    \includegraphics[width=\textwidth, trim={0 0 0 0}, clip]{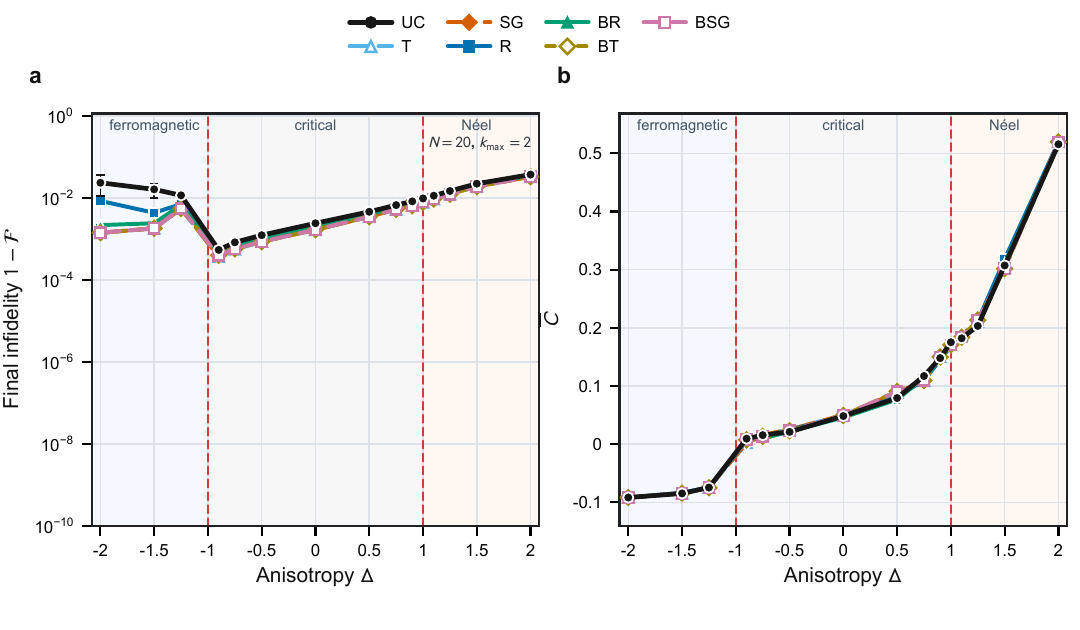}
    \caption{Fidelity and correlation diagnostics across the XXZ anisotropy sweep. (a) Final infidelity $1-\mathcal F$ and (b) distance-averaged correlation strength $\overline C$. The symmetry-constrained curves remain indistinguishable from the baseline, confirming no systematic loss of wavefunction accuracy or correlation structure.}
    \label{fig:app-xxz-nowarm-fid-corr-sweep}
\end{figure*}

\begin{figure*}[t]
    \centering
    \includegraphics[width=\textwidth, trim={0 0 0 0}, clip]{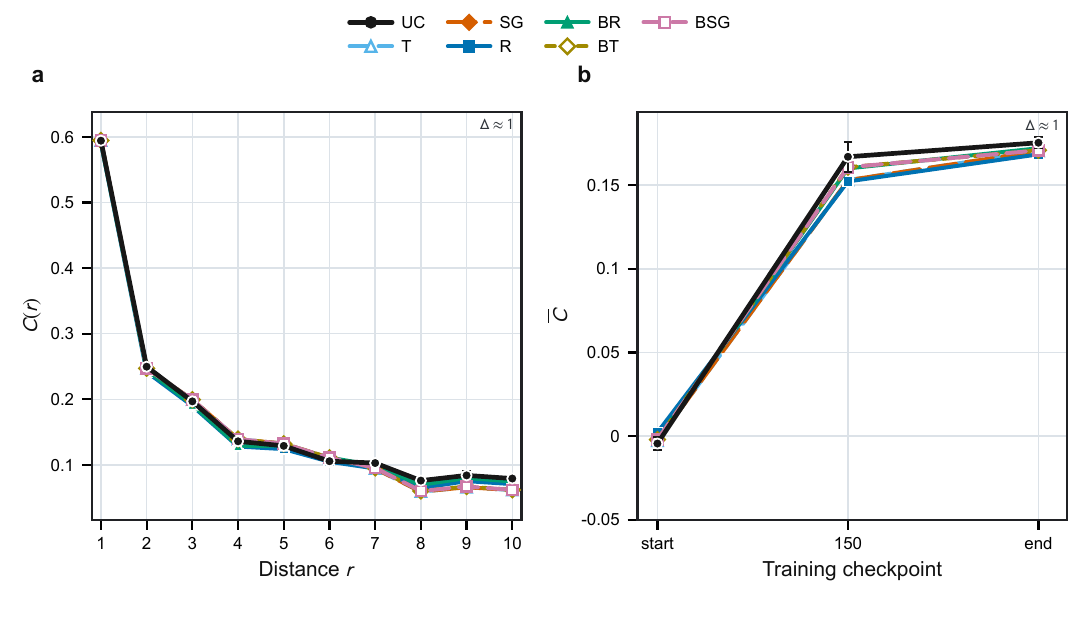}
    \caption{Correlation diagnostics near the XXZ isotropic point ($\Delta\simeq1$). (a) The endpoint spatial profile $C(r)$ and (b) the training-time evolution of $\overline C$ show identical behavior across all ansätze. The compressed manifolds retain the full capacity to build the physical correlation structure during training.}
    \label{fig:app-xxz-nowarm-corr-profile-growth}
\end{figure*}

This performance strictly translates to the XXZ anisotropy sweep (Figs.~\ref{fig:app-xxz-nowarm-fid-corr-sweep} and~\ref{fig:app-xxz-nowarm-corr-profile-growth}). The exact match in $1-\mathcal F$ across varying $\Delta$ proves that the symmetry constraints do not bottleneck global wavefunction accuracy. Resolving the correlations near the isotropic point ($\Delta \simeq 1$) shows that the full spatial decay profile $C(r)$ and its dynamical growth during training are fully preserved. Across both models, parameter tying successfully excises structurally redundant degrees of freedom without degrading the physical fidelity or the real-space observable structure of the target state.

\section{Large-System Sampled Geometry Proxies}
\label{app:sampled-geometry-proxies}

Computing the dense Jacobian, exact Hessian, or full Fubini--Study (FS) matrix is intractable for large-scale systems. To probe the endpoint optimization geometry in this regime, we employ scalable Monte Carlo proxies derived directly from the optimized variational state. For a normalized state $\psi_\theta(\sigma)$ parameterized by real variables $\theta_a$, we sample the log-derivative observables $O_a(\sigma)=\partial_{\theta_a}\log \psi_\theta(\sigma)$ from $|\psi_\theta(\sigma)|^2$. The exact FS metric is the real part of the covariance matrix,
\begin{equation}
    S_{ab}=\mathrm{Re}\left[\langle O_a^*O_b\rangle-\langle O_a^*\rangle\langle O_b\rangle\right].
\end{equation}

Rather than explicitly diagonalizing $S$, we bypass dense matrix operations by computing three scalar proxies for the active tangent-space dimension. We define the active set of diagonal entries $\mathcal{A} = \{a : S_{aa} > \tau\}$ for a numerical rank tolerance $\tau$, yielding the diagonal active-direction count $r_{\mathrm{diag}} = |\mathcal{A}|$. The effective dimension is further quantified by the entropy rank of the positive diagonal entries,
\begin{equation}
    r_{\mathrm{eff}}^{\mathrm{diag}}=\exp\left[-\sum_a p_a\log p_a\right],\quad p_a=\frac{S_{aa}}{\sum_b S_{bb}}.
\end{equation}
To capture the global conditioning without full diagonalization, we estimate the stable rank via randomized trace and Frobenius-norm probes,
\begin{equation}
    r_{\mathrm{stable}}=\frac{\mathrm{tr}(S)^2}{\|S\|_F^2}.
\end{equation}

To track the local training landscape, we define a local useful-volume fraction by sampling random perturbations $\delta\theta$ around the optimized endpoint $\theta_*$. This evaluates the fraction of perturbed states that remain within a strict energy tolerance $\epsilon$,
\begin{equation}
    f_{\mathrm{good}}^{\mathrm{MC}}=\frac{1}{M}\sum_{m=1}^{M}\mathbf 1\left[E(\theta_*+\delta\theta_m)-E(\theta_*)\leq \epsilon\right].
\end{equation}

\begin{figure*}[t]
    \centering
    \includegraphics[width=\textwidth, trim={0 0 0 0}, clip]{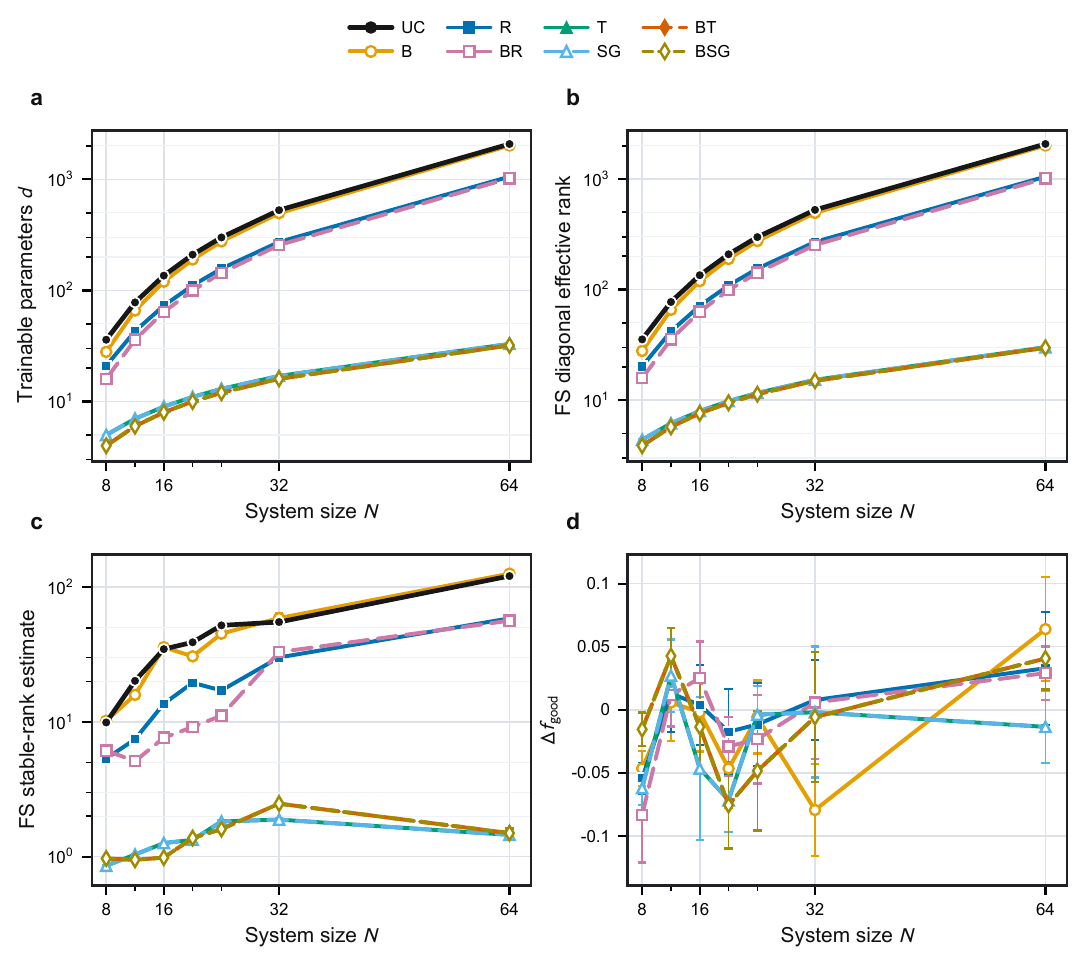}
    \caption{\textbf{Sampled large-system geometry proxies.} \textbf{a}, Trainable parameter count $d$ for all-to-all $k_{\max}=2$ diagonal NQS. \textbf{b}, FS effective-rank proxy ($r_{\mathrm{eff}}^{\mathrm{diag}}$). \textbf{c}, Randomized FS stable-rank estimate ($r_{\mathrm{stable}}$). \textbf{d}, Relative change in the Monte Carlo useful-volume fraction ($f_{\mathrm{good}}^{\mathrm{MC}}$) compared to unconstrained baselines. These scalable proxies bypass dense exact calculations while confirming geometric regularization at scale. Labels: UC (unconstrained); B (bitflip); R (reflection); BR (bitflip + reflection); T (translations); SG (space group); BT (bitflip + translations); BSG (bitflip + space group).}
    \label{fig:sampled-geometry-proxies}
\end{figure*}

As shown in Fig.~\ref{fig:sampled-geometry-proxies}, these scalable proxies directly reproduce the exact, small-system geometry scaling. The diagonal FS effective rank drops commensurately with the parameter count hierarchy, confirming that symmetry compilation analytically prunes fundamentally active tangent directions rather than merely redefining coordinate labels. Furthermore, the stable rank remains orders of magnitude smaller than the effective rank, revealing a highly anisotropic tangent space where metric weight is concentrated in a tight subset of directions. 

Crucially, the useful-volume fraction $f_{\mathrm{good}}^{\mathrm{MC}}$ remains strictly stable relative to the unconstrained baseline. If symmetry projection artificially overconstrained the variational family, this fraction would collapse. Instead, these large-system diagnostics are consistent with structural symmetry acting as a geometric regularizer by reducing the sampled tangent-space dimension without producing a systematic collapse of the locally acceptable fraction.

\section{Across-Seed and Along-Trajectory Robustness}
\label{app:geometry-training-robustness}

The main-text scaling analysis reveals a consistent geometric hierarchy among the weakly reduced $\{\mathrm{UC},\mathrm{B}\}$, reflection-reduced $\{\mathrm{R},\mathrm{BR}\}$, and strongly compressed spatial-symmetry families $\{\mathrm{T},\mathrm{SG},\mathrm{BT},\mathrm{BSG}\}$. To test whether this separation is robust to stochastic initialization and persists during optimization, we analyze the geometry across multiple seeds and training checkpoints. We evaluate the $N=64$ critical TFIM at $J=h=1$ and $k_{\max}=2$ using 50 independent random seeds for each of the eight ansatz families. Geometry diagnostics are extracted at initialization (epoch 0), the optimization midpoint (epoch 150), and the final iterate (epoch 300), yielding 1200 checkpoint evaluations from 400 independent optimization trajectories.

\begin{figure*}[t]
    \centering
    \includegraphics[width=\textwidth, trim={0 0 0 0}, clip]{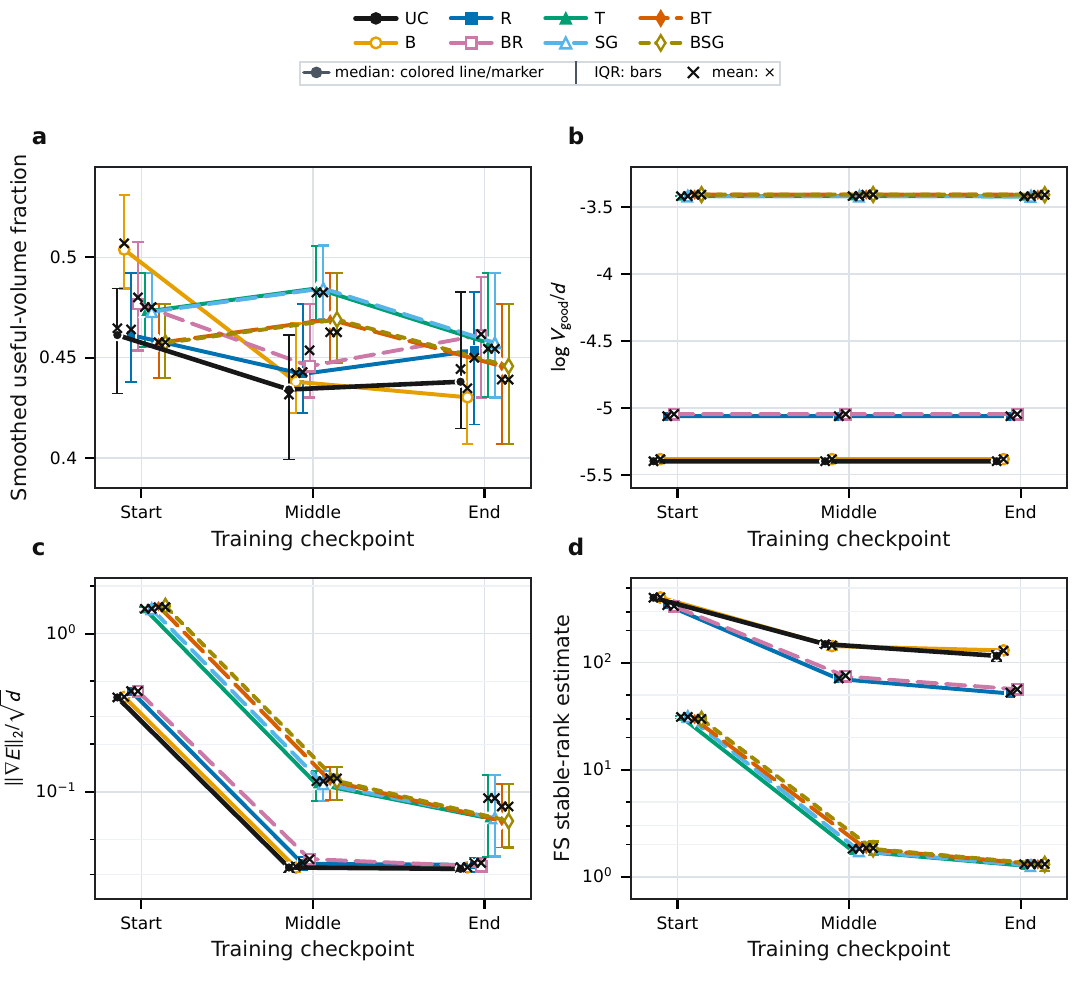}
    \caption{\textbf{Across-seed and along-trajectory geometry robustness at $N=64$.} Each ansatz family is trained over $50$ independent seeds. Colored markers and lines denote the median, vertical bars span the interquartile range, and black crosses indicate the arithmetic mean. Checkpoints correspond to initialization, training midpoint, and final iterate. Panels display (a) the smoothed useful-volume fraction, (b) the dimension-normalized log useful volume, (c) the dimension-normalized gradient norm $\lVert \nabla E\rVert_2/\sqrt{d}$, and (d) the randomized Fubini--Study stable rank. The rigid separation between the weakly reduced $\{\mathrm{UC},\mathrm{B}\}$, intermediate $\{\mathrm{R},\mathrm{BR}\}$, and strongly compressed $\{\mathrm{T},\mathrm{SG},\mathrm{BT},\mathrm{BSG}\}$ sectors persists uniformly across random seeds and optimization time, confirming the structural hierarchy identified in the system-size scaling limits. Labels: UC (unconstrained); B (bitflip); R (reflection); BR (bitflip + reflection); T (translations); SG (space group); BT (bitflip + translations); BSG (bitflip + space group).}
    \label{fig:geometry-training-robustness}
\end{figure*}

As shown in Fig.~\ref{fig:geometry-training-robustness}, the trajectory-resolved data reproduce the broad hierarchy observed in the system-size analysis. First, the randomized Fubini--Study stable rank [Fig.~\ref{fig:geometry-training-robustness}(d)] maintains the established symmetry hierarchy from initialization through convergence. Although the absolute rank decreases for all families, the separation between the symmetry sectors remains visible throughout optimization. Second, the dimension-normalized log useful volume $\log V_{\mathrm{good}}/d$ [Fig.~\ref{fig:geometry-training-robustness}(b)] exhibits identically strong partitioning. The strict temporal persistence of this metric demonstrates that the local useful volume is fundamentally dictated by the static symmetry constraints defining the manifold, not merely generated dynamically upon reaching a specific ground state.

Crucially, this geometric persistence occurs alongside genuine parameter evolution. The dimension-normalized gradient norm $\lVert \nabla E\rVert_2/\sqrt{d}$ [Fig.~\ref{fig:geometry-training-robustness}(c)] drops by over an order of magnitude from initialization to the midpoint across all families. Consequently, the surviving hierarchy reflects persistent geometric differences between the constrained manifolds along active and structurally distinct optimization trajectories.

Finally, the sampled estimates exhibit stable summary statistics across the evaluated seeds and checkpoints. To quantify the stability of the Monte Carlo volume estimator [Fig.~\ref{fig:geometry-training-robustness}(a)], we define the normalized mean--median displacement
\begin{equation}
    D_{s,t} = \frac{\left| \overline{x}_{s,t} - \widetilde{x}_{s,t} \right|}{Q_{75}^{(s,t)} - Q_{25}^{(s,t)}},
\end{equation}
where $s$ and $t$ index the ansätze family and training checkpoint, respectively. Across all $24$ evaluated scenario--checkpoint cells, the median displacement $D_{s,t}$ lies between $0.06$ and $0.09$, with a global maximum below $0.29$. The arithmetic mean lies within the interquartile range in every evaluated cell. These ensemble statistics support the robustness of the reported geometric hierarchy across the sampled seeds and checkpoints and are consistent with a structural contribution from the imposed symmetry constraints. The arithmetic mean universally resides deep within the interquartile range. These ensemble statistics prove that the geometric compression mechanisms reported in the main text are strictly reproducible, trajectory-independent, and structurally inherent to the symmetry-projected variational spaces.

\section{Diagnostic decomposition of the target-aware geometric results}
\label{app:geometry-diagnostics}

The main text summarizes the target-aware variational geometry through two complementary quantities. Definition~\ref{def:trainability_metric} measures the global concentration of target-accurate states through the dimensionless useful-volume fraction $f_{\epsilon}$, while Eq.~\eqref{eq:characteristic_basin_scale} measures the characteristic local scale $R_{\epsilon}$ obtained by converting the volume of the same accurate region into a volume-equivalent linear quantity. The purpose of this Appendix is to offer supplementart diagnostics for the intermediate geometric ingredients entering these two quantities and rule out simpler explanations of the hierarchy reported in Figs.~\ref{fig:useful_expressibility_epsilon} and~\ref{fig:characteristic_basin_width_epsilon}.

For any diagnostic quantity $X$ and symmetry sector $s$, we use the unconstrained-relative difference
\begin{equation}\label{eq:appendix_uc_difference}\Delta_{\rm UC}X_s(N)\equiv X_s(N)-X_{\rm UC}(N).\end{equation}
A positive value therefore indicates an increase relative to the unconstrained ansatz at the same system size. These differences are used only to expose the components entering the main geometric measures. They are not interpreted as wall-time speedups, optimizer success probabilities, or counts of solutions.

\subsection{Decomposition of the local positive-curvature geometry}
\label{app:local-hessian-diagnostics}

The first diagnostic resolves the two spectral contributions entering Theorem~\ref{thm:good_volume}. At a stationary physical minimum, the local good volume depends on both the positive-curvature rank $r_+$ and the physical pseudodeterminant $\det_+K_*$. The former determines the dimension and tolerance exponent of the locally confined region, while the latter determines its aggregate stiffness. Reporting only one of these quantities would therefore be insufficient to determine whether symmetry compilation produces a genuinely broader basin within the physical directions that remain.

\begin{figure*}[t]
    \centering
    \includegraphics[width=\textwidth]{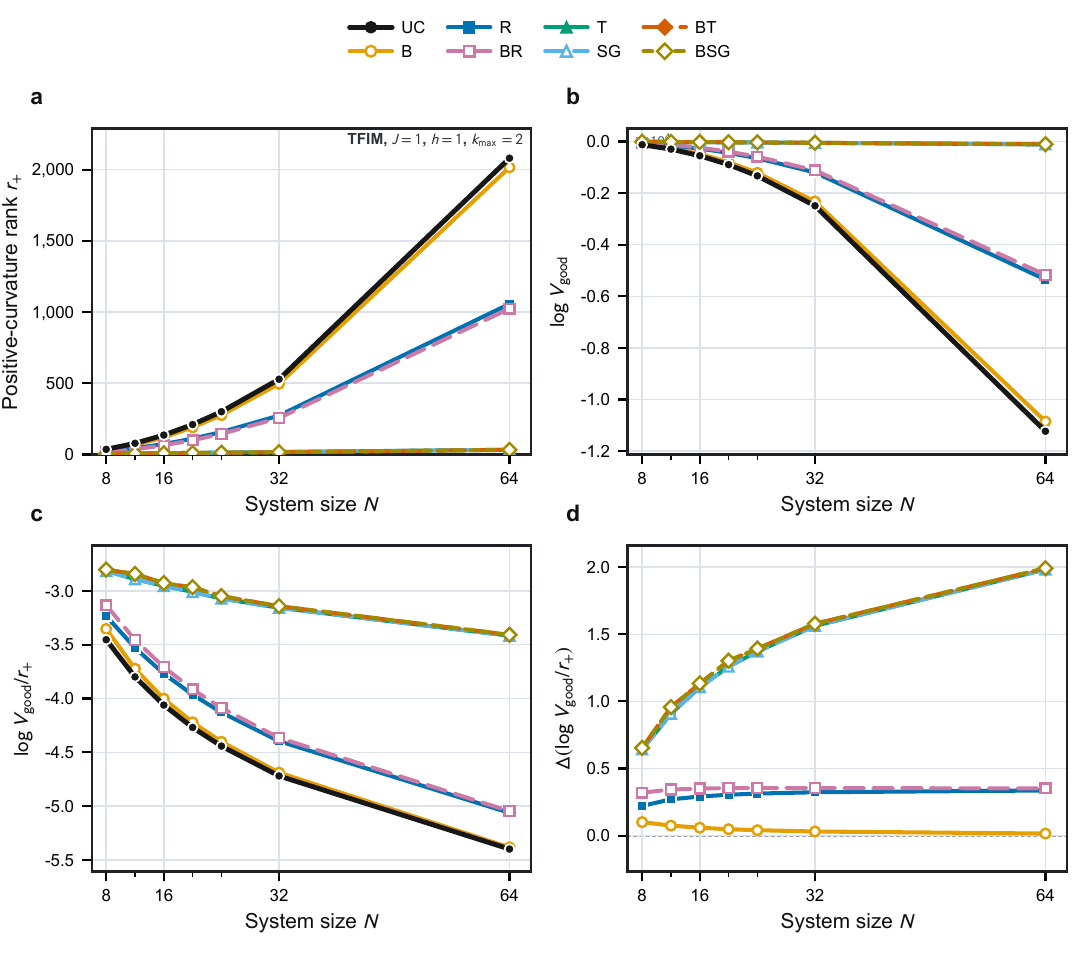}
    \caption{\textbf{Spectral decomposition of the local target basin.} \textbf{a}, Positive-curvature rank $r_+$ obtained from the resolved positive eigenvalues of the Fubini--Study-normalized curvature operator $K_*$ defined in Eq.~\eqref{eq:normalized_curvature}. \textbf{b}, Logarithm of the local good volume $V_{{\rm good},+}^{\rm FS}$ appearing in Theorem~\ref{thm:good_volume}. \textbf{c}, Rank-normalized log volume $\log V_{{\rm good},+}^{\rm FS}/r_+=\log R_\epsilon$, where $R_\epsilon$ is the characteristic basin scale defined in Eq.~\eqref{eq:characteristic_basin_scale}. \textbf{d}, Increase $\Delta_{\rm UC}\log R_\epsilon$ relative to the unconstrained ansatz at the same system size. Results are shown for the critical TFIM with $J=h=1$ and $k_{\max}=2$. Labels: UC (unconstrained); B (bitflip); R (reflection); BR (bitflip + reflection); T (translations); SG (space group); BT (bitflip + translations); BSG (bitflip + space group).}
    \label{fig:app-local-hessian-diagnostics}
\end{figure*}

Figure~\ref{fig:app-local-hessian-diagnostics}a verifies that the reduction produced by spatial symmetry compilation extends to the target-confined physical dimension $r_+$ and is not limited to the nominal number of trainable coordinates. The unconstrained and bit-flip-only ansatz develop a rapidly increasing number of resolved confining directions, while translation-containing constraints retain a much smaller positive-curvature subspace. This behavior is consistent with the hierarchy $r_+\leq r_{\rm uncon/con}\leq|\vec\theta|$ established in Eq.~\eqref{eq:dimension_hierarchy}. It shows that orbit tying changes the dimension of the locally confined physical basin rather than merely relabeling the original coordinate space.

Panel~b displays the corresponding absolute local volume from Theorem~\ref{thm:good_volume}. The large separation between symmetry sectors is expected because $V_{{\rm good},+}^{\rm FS}$ contains the explicit factor $(2\epsilon)^{r_+/2}$. Consequently, the raw volume in panel~b must not be interpreted by itself as evidence that the retained directions are broader. Ansätze with different $r_+$ assign volumes to spaces of different intrinsic dimension, and the tolerance factor can dominate their ratio.

Panels~c and d provide the rank-aware diagnostic required to interpret the main geometric results. By Eq.~\eqref{eq:characteristic_basin_scale}, $\log V_{{\rm good},+}^{\rm FS}/r_+$ is exactly $\log R_{\epsilon}$. Taking the $r_+$th root converts the local good volume into a volume-equivalent linear scale and permits comparisons among ansätze with different positive-curvature ranks. The translation-containing sectors retain substantially larger values of $R_{\epsilon}$ than the unconstrained and bit-flip-only sectors, and their separation from the unconstrained ansatz increases with system size. The larger useful-volume fractions reported in Fig.~\ref{fig:useful_expressibility_epsilon} therefore cannot be attributed solely to a reduction in $r_+$. The retained confining directions also support a larger characteristic basin scale, consistent with the direct results in Fig.~\ref{fig:characteristic_basin_width_epsilon}.

This diagnostic establishes the local part of the main claim. It does not determine the global concentration $f_\epsilon$, because $R_\epsilon$ contains no information about the total reachable volume $\operatorname{Vol}_{\rm FS}(\mathcal M_\Theta)$ appearing in Definition~\ref{def:trainability_metric}. The global and local statements become simultaneous only when Fig.~\ref{fig:app-local-hessian-diagnostics} is read together with the main useful-expressibility figure.

\subsection{Compression of the physical tangent space}
\label{app:fs-tangent-diagnostics}

The next diagnostic tests whether the parameter reduction survives projection onto physically distinguishable state variations. The pullback Fubini--Study metric $S=\operatorname{Re}[J^\dagger\Pi_\perp J]$ defined in Eq.~\eqref{eq:defn_S_metric} vanishes along parameter directions that do not generate a first-order physical displacement. Its resolved rank therefore measures the local physical tangent dimension rather than the raw coordinate dimension~\cite{haug2021capacity,dash2025efficiency}.

\begin{figure*}[t]
    \centering
    \includegraphics[width=\textwidth]{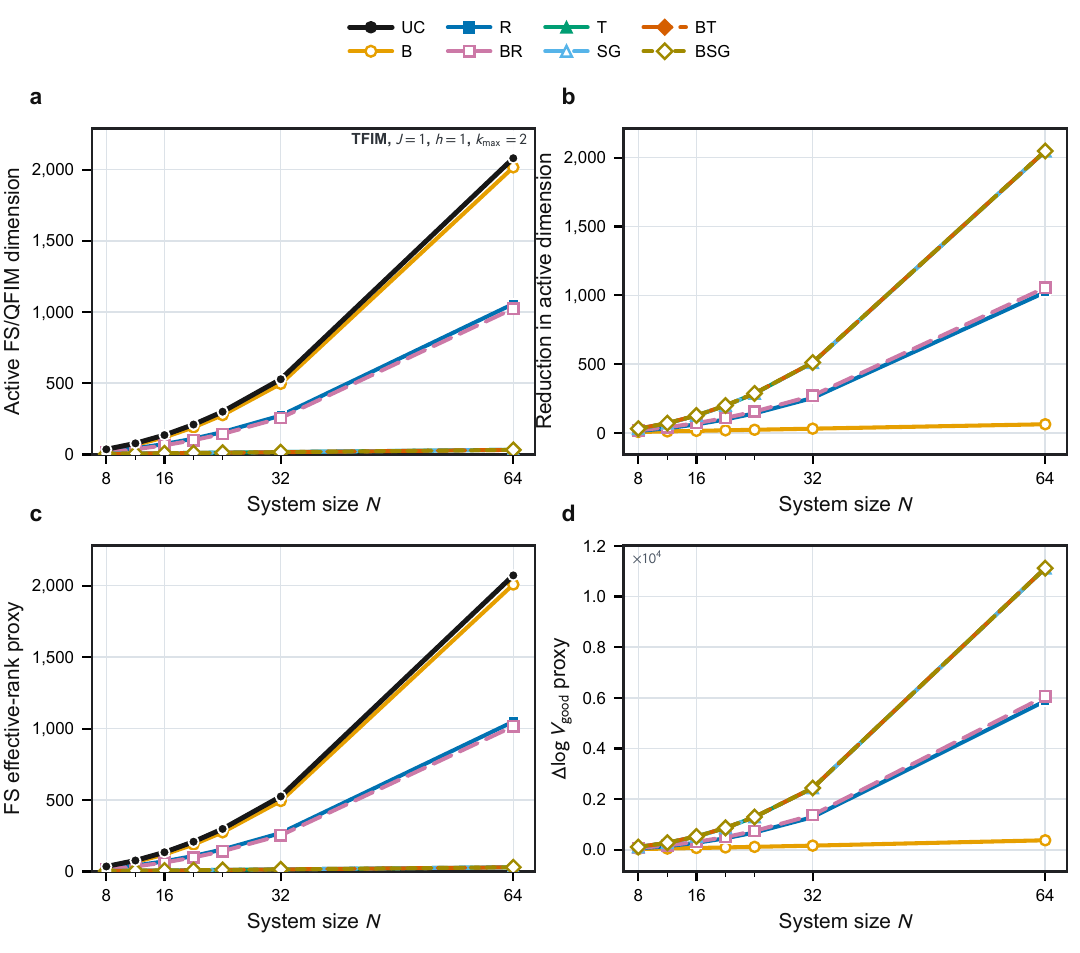}
    \caption{\textbf{Information-geometric diagnostics of the reachable tangent space.} \textbf{a}, Resolved active dimension of the Fubini--Study metric, estimating the rank of the physical tangent space defined by Eq.~\eqref{eq:defn_S_metric}. \textbf{b}, Reduction in active dimension relative to the unconstrained ansatz at the same system size. \textbf{c}, Spectral effective-rank proxy for the Fubini--Study metric. \textbf{d}, Unconstrained-relative change in the sampled local log-volume proxy $\Delta_{\rm UC}\log\widehat V_{\rm good}$. The quantities in panels c and d are scalable spectral or sampling proxies and are not substituted for the exact definitions of $f_\epsilon$ or $R_\epsilon$. Results are shown for the critical TFIM with $J=h=1$ and $k_{\max}=2$. Labels: UC (unconstrained); B (bitflip); R (reflection); BR (bitflip + reflection); T (translations); SG (space group); BT (bitflip + translations); BSG (bitflip + space group).}
    \label{fig:app-fs-tangent-diagnostics}
\end{figure*}

Figure~\ref{fig:app-fs-tangent-diagnostics}a shows that the symmetry hierarchy persists after the coordinate Jacobian is projected onto the physical tangent space. The active Fubini--Study dimension grows most rapidly for the unconstrained and bit-flip-only ansätze, is reduced by reflection-based constraints, and remains smallest under translation-containing constraints. Panel~b makes the corresponding removal of active physical directions explicit. The agreement between the coordinate compression and the Fubini--Study rank reduction rules out the interpretation that symmetry compilation merely replaces a large coordinate representation with an equally large physical manifold written in different variables.

Panel~c provides a complementary spectral diagnostic. The effective-rank proxy provides a continuous measure of the number of appreciably weighted spectral directions rather than counting every mode above a fixed numerical threshold~\cite{roy2007effective}. It reproduces the same broad separation between unconstrained, reflection-based, and translation-containing sectors. The agreement between thresholded and spectrally weighted notions of dimension shows that the observed hierarchy is not produced only by a particular rank cutoff.

Panel~d propagates the sampled tangent-space information into a local volume proxy. Its increase under stronger spatial constraints is consistent with the mechanism exposed by Theorem~\ref{thm:good_volume}, since reducing the number of active directions and changing their spectral weights both affect the estimated local volume. This panel is nevertheless a mechanistic diagnostic rather than a direct evaluation of Definition~\ref{def:trainability_metric}. It does not contain the full global Fubini--Study denominator in Eq.~\eqref{eq:total_fs_volume}, and it does not replace the rank-normalized basin comparison supplied by $R_\epsilon$.

The conclusion supported by Fig.~\ref{fig:app-fs-tangent-diagnostics} is therefore deliberately limited. Symmetry compilation reduces the physically active tangent dimension and changes the distribution of Fubini--Study spectral weight. The stronger conclusion that this compressed physical manifold is more efficiently concentrated around the target follows from the normalized $f_\epsilon$ values reported in the main text.

\subsection{Stability of the geometric diagnostics during optimization}
\label{app:training-checkpoint-diagnostics}

The construction in Theorem~\ref{thm:good_volume} is local and assumes a stationary physical minimum with no resolved negative-curvature modes. We therefore examine the geometric estimators at multiple points along the optimization trajectory. This test determines whether the endpoint values used in the main text are part of a stable late-training regime or arise from an isolated terminal fluctuation.

\begin{figure*}[t]
    \centering
    \includegraphics[width=\textwidth]{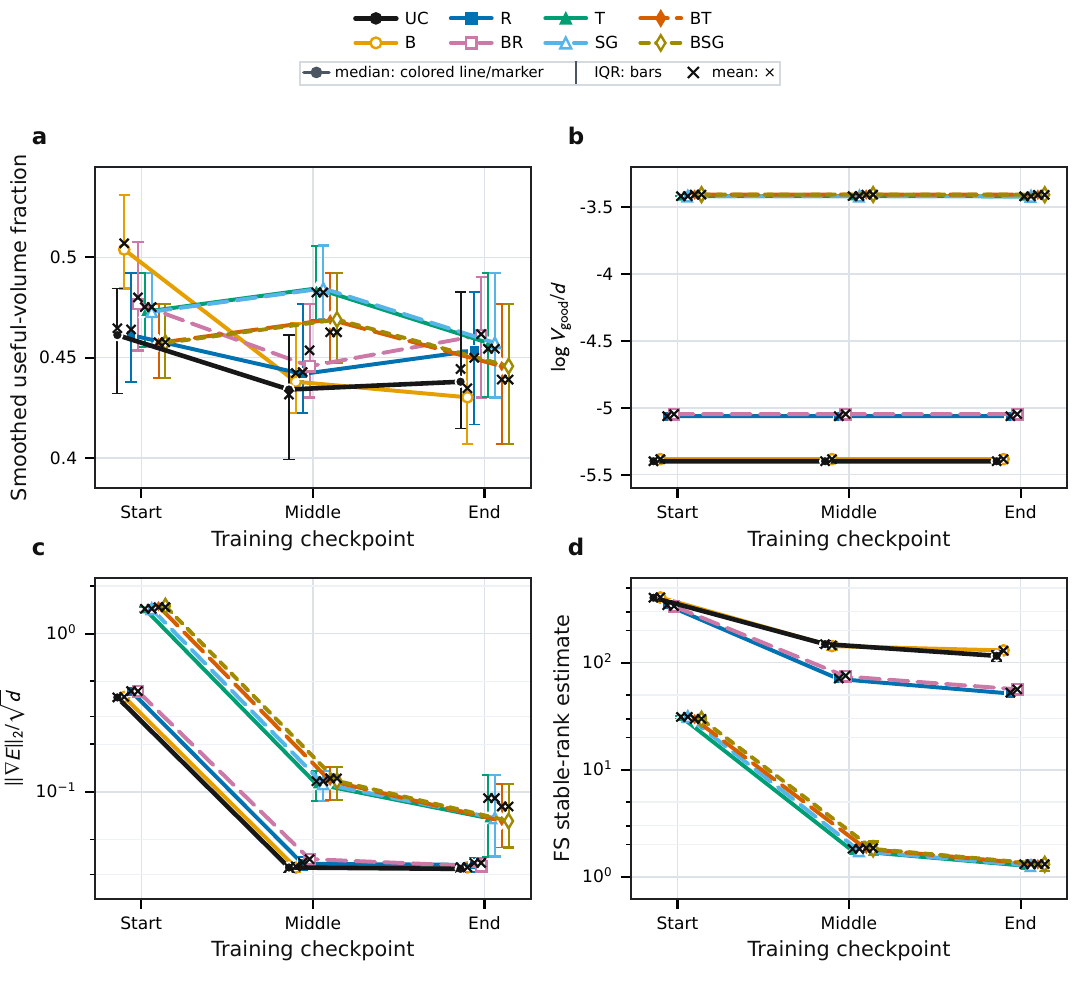}
    \caption{\textbf{Evolution of geometric diagnostics during training.} \textbf{a}, Smoothed sampled good-volume fraction at the start, middle, and end of optimization. \textbf{b}, Local log-volume proxy normalized by the raw trainable dimension, $\log\widehat V_{\rm good}/d$. This panel is a trajectory-stability diagnostic and must not be identified with $\log R_\epsilon$, whose normalization uses $r_+$ rather than $d$. \textbf{c}, Energy-gradient norm per square root of the trainable dimension, $\|\nabla E\|_2/\sqrt d$. \textbf{d}, Fubini--Study stable-rank estimate. Colored markers and lines show medians, vertical bars show interquartile ranges, and crosses show means over seeds. Labels: UC (unconstrained); B (bitflip); R (reflection); BR (bitflip + reflection); T (translations); SG (space group); BT (bitflip + translations); BSG (bitflip + space group).}
    \label{fig:app-training-checkpoint-diagnostics}
\end{figure*}

The reduction of $\|\nabla E\|_2/\sqrt d$ in Fig.~\ref{fig:app-training-checkpoint-diagnostics}c confirms that the middle and final checkpoints lie substantially closer to stationarity than the initial configurations. This behavior supports the use of endpoint Hessian geometry in Theorem~\ref{thm:good_volume}. A small sampled gradient does not mathematically prove exact stationarity, so this panel should be understood as a numerical consistency check rather than as a replacement for the stationary-point assumption.

Panel~d shows that the Fubini--Study spectrum changes substantially during optimization. The stable-rank estimate falls between the initial and middle checkpoints and then changes more gradually toward the endpoint. This observation justifies evaluating the metric and curvature at the optimized state rather than assuming that the geometry of the initialization remains representative throughout training. It also shows that the separation between symmetry sectors is already present before the final iteration and is not generated by a discontinuous endpoint operation.

\begin{figure*}[t]
    \centering
    \includegraphics[width=\textwidth]{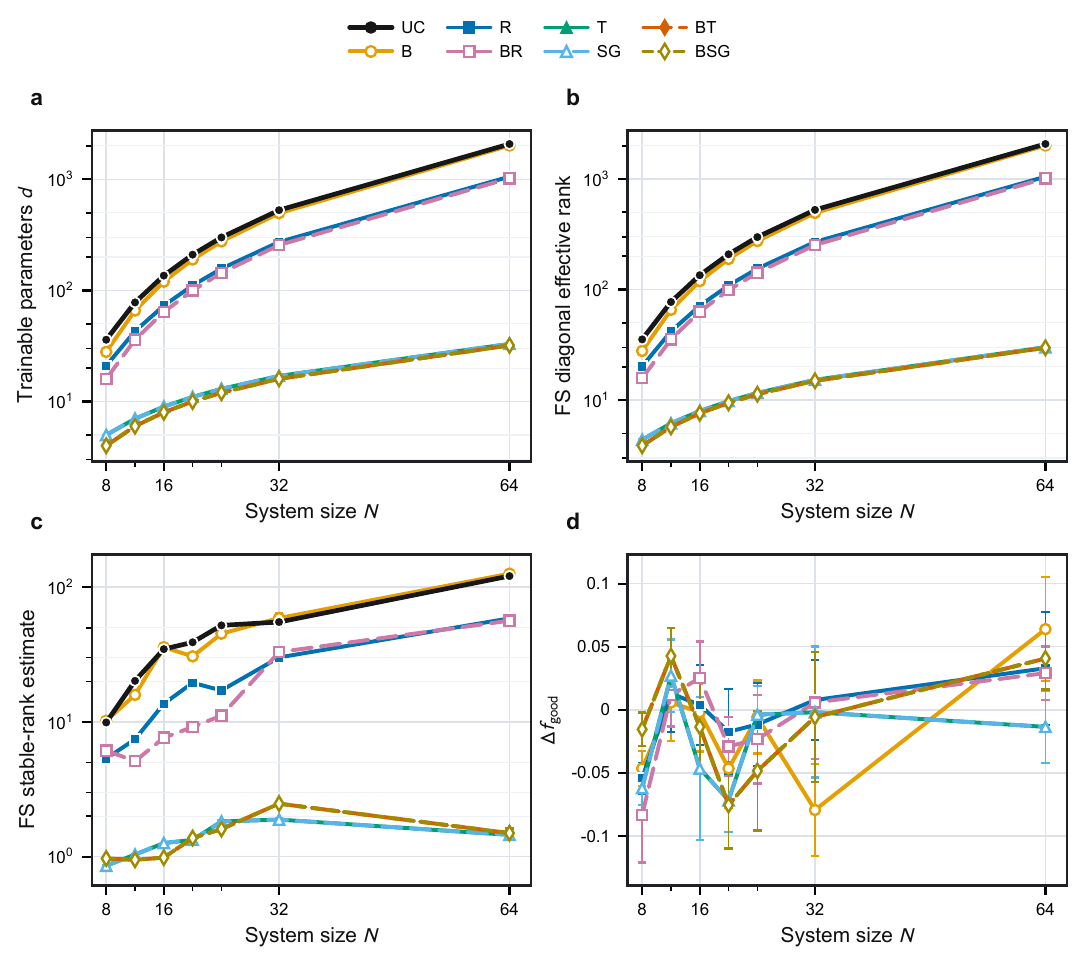}
    \caption{\textbf{Scaling and consistency of sampled Fubini--Study proxies.} \textbf{a}, Number of trainable parameters $d$. \textbf{b}, Diagonal effective-rank estimate obtained from the sampled Fubini--Study metric. \textbf{c}, Fubini--Study stable-rank estimate, which weights the tangent spectrum according to the number of appreciably contributing modes. \textbf{d}, Difference $\Delta_{\rm UC}\widehat f_{\rm good}$ between the sampled local good fraction of each constrained ansatz and the unconstrained ansatz at the same system size. Error bars report the sampling uncertainty over seeds. The quantity $\widehat f_{\rm good}$ is a local Monte Carlo diagnostic and is distinct from the globally normalized $f_\epsilon$ in Definition~\ref{def:trainability_metric}. Labels: UC (unconstrained); B (bitflip); R (reflection); BR (bitflip + reflection); T (translations); SG (space group); BT (bitflip + translations); BSG (bitflip + space group). Labels: UC (unconstrained); B (bitflip); R (reflection); BR (bitflip + reflection); T (translations); SG (space group); BT (bitflip + translations); BSG (bitflip + space group).}
    \label{fig:app-sampled-fs-proxies}
\end{figure*}

The sampled good-fraction and coordinate-normalized local-volume diagnostics in panels a and b do not exhibit an abrupt terminal jump. Their endpoint values instead continue the behavior observed at the intermediate checkpoint. This continuity is the relevant diagnostic conclusion. Panel~b is intentionally not used as a cross-ansatz basin-width measure because division by $d$ does not implement the physical normalization in Eq.~\eqref{eq:characteristic_basin_scale}. The main comparison across symmetry sectors remains $\log R_\epsilon=\log V_{{\rm good},+}^{\rm FS}/r_+$, as reported in Fig.~\ref{fig:app-local-hessian-diagnostics}c and in the main text.

For the strict application of Theorem~\ref{thm:good_volume}, the endpoint spectrum must additionally contain no resolved negative-curvature modes. Only endpoints satisfying this criterion should enter the reported local-volume and characteristic-scale statistics.

\subsection{Scalable proxy and estimator diagnostics}
\label{app:sampled-proxy-diagnostics}

The exact construction, storage, and spectral decomposition of the Fubini--Study metric and curvature operators become increasingly expensive as the number of variational coordinates grows, motivating scalable stochastic or reduced spectral estimators~\cite{RBM_Geom_learning,dash2025efficiency}. We therefore compare several lower-cost estimators of the physical dimension and sampled local good fraction. These quantities are used to test the scalability and numerical consistency of the geometric pipeline. Therefore, they are not promoted to independent definitions of target-aware useful expressibility.

Panels~a and b of Fig.~\ref{fig:app-sampled-fs-proxies} show that the diagonal Fubini--Study estimator recovers the expected compression hierarchy. The unconstrained and bit-flip-only sectors retain the largest effective dimensions, reflection-based sectors form an intermediate group, and translation-containing constraints retain only a small orbit-count-controlled set of directions. The close correspondence with the trainable dimension confirms that the sampled diagonal metric detects the structural orbit tying implemented in the parameterization.

The stable-rank estimate in panel~c is substantially smaller than either the raw coordinate count or the diagonal effective rank. This difference is expected because the stable rank weights the full spectrum continuously relative to its dominant mode, so directions carrying negligible spectral weight contribute only weakly~\cite{ipsen2025stable}. Its importance is diagnostic. Although the numerical magnitude depends on the estimator, the same three broad symmetry regimes remain visible. The physical compression hierarchy is therefore not an artifact of counting all diagonal entries as equally active.

Panel~d serves as a negative-control test. Within the reported uncertainties, the sampled local good fraction is generally comparable to the unconstrained baseline even after severe spatial compression. The figure therefore shows that the scalable sampling procedure does not detect a systematic collapse of the locally acceptable fraction when translation or space-group constraints are imposed. It should not be described as reproducing the large global enhancement of $f_\epsilon$ in the main text. The quantity plotted here is a sampled local fraction, whereas Definition~\ref{def:trainability_metric} normalizes the connected target-accurate Fubini--Study volume by the total reachable physical volume of the ansatz. The strong enhancement in the main figure arises only after this global denominator and the local curvature construction are placed on the same physical footing.

The four diagnostics establish a consistent chain of evidence behind the main target-aware geometric result. First, the reduction in trainable parameters survives projection onto the Fubini--Study tangent space, showing that symmetry compilation changes the physically reachable local manifold rather than only its coordinate representation. Second, the local Hessian decomposition shows that the increase in the good-volume fraction is not generated solely by the smaller positive-curvature rank. The characteristic scale $R_\epsilon$ remains larger for translation-containing constraints after the good volume is converted to a rank-normalized linear quantity. Third, the checkpoint analysis shows that the endpoint geometry is reached continuously as the gradient decreases and is not produced by an isolated terminal fluctuation. Fourth, diagonal-rank, stable-rank, and sampled-volume proxies preserve the same broad structural hierarchy and do not reveal a systematic loss of locally acceptable configurations under strong symmetry compression.

\end{document}